\documentclass[a4paper,UKenglish,cleveref, autoref, thm-restate]{lipics-v2021}
\usepackage{float}
\usepackage{bbding}
\usepackage{stackengine}
\usepackage{todonotes}
\usepackage{subcaption}
\usepackage{xspace}
\usepackage{todonotes}
\DeclareMathOperator{\lca}{\mathrm{lca}}
\DeclareMathOperator{\fr}{\mathrm{fr}}

\title{Partial and Simultaneous Transitive Orientations via Modular Decompositions} 
 
\titlerunning{Partial and Simultaneous Transitive Orientations}

\author{Miriam Münch}{Faculty of Computer Science and Mathematics, University of
Passau, Germany}{muenchm@fim.uni-passau.de}{https://orcid.org/0000-0002-6997-8774}{}

\author{Ignaz Rutter}{Faculty of Computer Science and Mathematics, University of
Passau, Germany}{rutter@fim.uni-passau.de}{https://orcid.org/0000-0002-3794-4406}{}

\author{Peter Stumpf}{Faculty of Computer Science and Mathematics, University of
Passau, Germany}{stumpf@fim.uni-passau.de}{https://orcid.org/0000-0003-0531-9769}{}

\authorrunning{M. Münch, I. Rutter and P. Stumpf} 

\Copyright{Miriam Münch, Ignaz Rutter and Peter Stumpf} 

\ccsdesc[100]{Theory of Computation $\to$ Design and analysis of algorithms}%

\keywords{representation extension, simultaneous representation, comparability graph, permutation graph, circular permutation graph, modular decomposition} %

\relatedversion{}

\funding{This work was supported by grant RU$\,$1903/3-1 of the German Research Foundation (DFG).}

\newcommand\perm{\textsc{Perm}\xspace}
\newcommand{\cperm}{\textsc{CPerm}\xspace}
\newcommand{\comp}{\textsc{Comp}\xspace}
\newcommand{\repext}{\textsc{RepExt}\xspace}
\newcommand{\orext}{\textsc{OrientExt}\xspace}
\newcommand{\simul}{\textsc{SimRep}\xspace}
\newcommand{\sunflower}{\textsc{SimRep$^\star$}\xspace} 
\newcommand{\simorNon}{\textsc{SimOrient}\xspace}
\newcommand{\simor}{\textsc{SimOrient$^\star$}\xspace}
\newcommand{\rep}[1]{\ensuremath {\mathrm{rep}(#1)}\xspace}
\newcommand{\repS}[1]{\ensuremath {\mathrm{rep}_{T}(#1)}\xspace} 
\newcommand{\repB}[1]{\ensuremath {\mathrm{rep}_B(#1)}\xspace}
\newcommand{\repT}[1]{\ensuremath {\mathrm{rep}_T(#1)}\xspace}
\newcommand{\repMu}[1]{\ensuremath {\mathrm{rep}_\mu(#1)}\xspace}
\newcommand{\repNode}[2]{\ensuremath {\mathrm{rep}_{#1}(#2)}\xspace}
\newcommand\TO{\mathrm{to}\xspace}

\begin{document}
\nolinenumbers
\hideLIPIcs
\maketitle

\begin{abstract}
A natural generalization of the recognition problem for a geometric graph class
is the problem of extending a representation of a subgraph to a representation
of the whole graph.
A related problem is to find representations for multiple input graphs that coincide on subgraphs shared by the input graphs. A common restriction is the sunflower case where the
shared graph is the same for each pair of input graphs.
These problems translate to the setting of comparability graphs where the representations
correspond to transitive orientations of their edges.
We use modular decompositions to improve the runtime for the
orientation extension problem and the sunflower orientation problem to linear time.
We apply these results to improve the runtime for the partial representation problem and the
sunflower case of the simultaneous representation problem for permutation graphs to linear time. We also give the first efficient algorithms for these problems on circular permutation graphs. 
\end{abstract}

\section{Introduction}
\label{sec:introduction}

\begin{figure}[b]
  \centering
  \includegraphics{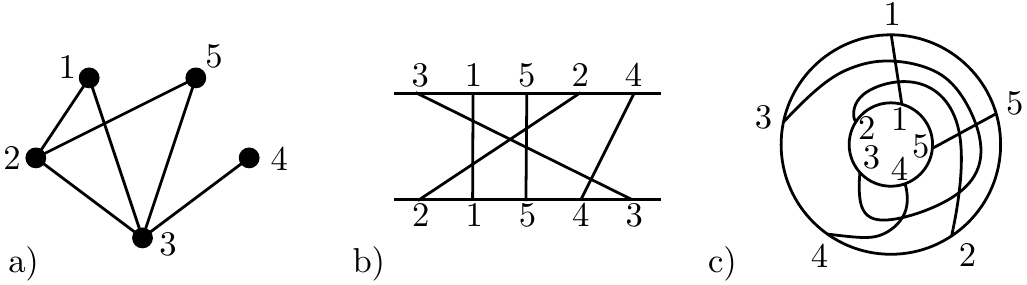}
  \caption{(a) a graph $G$ with (b) a permutation diagram and (c) a circular permutation diagram.}
  \label{fig:exampleCPG}
\end{figure}

Representations and drawings of graphs have been considered since graphs have
been studied~\cite{klavik2017extending}.
A \emph{geometric intersection representation} of a graph $G=(V,E)$ with
regards to a class of geometric objects~$\mathcal C$, 
is a map $R\colon V\to \mathcal C$ that assigns objects of~$\mathcal C$ to the
vertices of~$G$ such that $G$ contains an edge $uv$ if and only if the
intersection~$R(u)\cap R(v)$ is non-empty.
In this way, the class $\mathcal C$ gives rise to a class of graphs, namely the
graphs that admit such a representation.
As an example, consider \emph{permutation diagrams} where $\mathcal C$ consists
of segments connecting two parallel lines $\ell_1$, $\ell_2$, see
Figure~\ref{fig:exampleCPG}b, which defines the class \perm of \emph{permutation
graphs}.
Similarly, the class \cperm of \emph{circular permutation graphs}
is obtained by replacing $\ell_1$, $\ell_2$ with concentric circles and the
geometric objects with curves from $\ell_1$ to $\ell_2$ that pairwise intersect at most
once; see Figure~\ref{fig:exampleCPG}c.

A key problem in this context is the \emph{recognition problem}, which asks
whether a given graph admits such a representation for a fixed class~$\mathcal C$.
Klavík et al.~introduced the \emph{partial representation extension problem}
($\repext(\mathcal C)$) for intersection graphs where a representation $R'$ is
given for a subset of vertices $W\subseteq V$ and the question is whether~$R'$
can be extended to a representation $R$ of~$G$, in the sense that \mbox{$R|_W
=R'$~\cite{klavik2017extending}}.
They showed that \repext can be solved in linear time for interval graphs.
The problem has further been studied for proper/unit interval
graphs~\cite{klavik2017extendingPUI}, function and permutation
graphs (\perm)~\cite{klavik2012extending}, circle graphs~\cite{CFK13}, chordal
graphs~\cite{KKOS15}, and trapezoid graphs~\cite{KrawczykW17}.
Related extension problems have also been considered, e.g., for planar
topological~\cite{adfjk-tppeg-15,jkr-kttpp-13} and
straight-line~\cite{p-epsld-06} drawings, for 1-planar
drawings~\cite{eghkn-ep1pd-20}, for contact
representations~\cite{chaplick2014contact}, and for rectangular
duals~\cite{abs-2102-02013}. 

A related problem is the \emph{simultaneous representation problem}
($\simul(\mathcal C)$) where input graphs $G_1\dots,G_r$ that may share
subgraphs are given and the question is whether they have representations
$R_1,\dots,R_r$ such that for $i,j\in\{1,\dots,r\}$ the shared graph $H=G_i\cap
G_j$ has the same representation in $R_i$ and $R_j$, i.e.,~$R_i|_{V(H)}=R_j|_{V(H)}$.
If more than two input graphs are allowed, usually the \emph{sunflower case}
($\sunflower(\mathcal C)$) is considered, where the shared graph $H=G_i\cap
G_j$ is the same for any $i\ne j\in \{1,\dots,r\}$.
I.e., here the question is whether~$H$ has a representation that can be simultaneously
extended to $G_1,\dots,G_r$.
Simultaneous representations were first studied in the context of planar
drawings~\cite{Handbook2013, brass2007simultaneous}, where
the goal is to embed each input graph without edge crossings while shared
subgraphs have the same induced embedding. Unsurprisingly, many variants are
NP-complete~\cite{Gassner2006, Schaefer2013,
	Angelini2014,Estrella-Balderrama2008}. 

  Motivated by applications in visualization of temporal relationships, and for
  overlapping social networks or schedules, DNA fragments of similar organisms
  and adjacent layers on a computer chip, Jampani and Lubiw introduced the
  problem \simul for intersection graphs~\cite{jampani2012simultaneous}.
They provided polynomial-time algorithms for two chordal graphs and for
$\sunflower(\perm)$.  They also showed that in general $\simul$ is NP-complete
for three or more chordal graphs.
The problem was also studied for interval
graphs~\cite{Jampani2010, Blasius:2015:SPA:2846106.2738054,bok2018note},
proper/unit interval graphs~\cite{rutter2019simultaneous},
circular-arc graphs~\cite{bok2018note} and circle graphs~\cite{CFK13}. 

Many of the considered graph classes are related to the class \comp of
\emph{comparability graphs}~\cite{golumbic2004algorithmic}.
An \emph{orientation}~$O$ of a graph $G=(V,E)$ assigns to each edge of $G$ a direction.
The orientation~$O$ is \emph{transitive} if~$uv,vw\in O$ implies $uw\in O$. 
A comparability graph is a graph for which there is a \emph{transitive
orientation}. 
A \emph{partial orientation} is an orientation of a (not necessariliy induced)
subgraph of $G$. 
Similar to \repext, \sunflower and \simul, the problems \orext, \simor and \simorNon for
comparability graphs ask for a transitive orientation of a graph that
extends a given partial orientation and for transitive orientations
that coincide on the shared graph, respectively. 
The key ingredient
for the $O(n^3)$ algorithm solving $\repext(\perm)$ by Klavík et
al.~\cite{klavik2012extending} is a polynomial-time solution for
\orext based on the transitive orientation algorithm by
Gilmore and Hoffman~\cite{gilmore1964characterization}.  Likewise, the $O(n^3)$ algorithm solving $\sunflower(\perm)$
by Jampani and Lubiw~\cite{jampani2012simultaneous} is based on a
polynomial-time algorithm for \simor based on the transitive orientation algorithm by
Golumbic~\cite{golumbic2004algorithmic}. 

\begin{table}[t]
  \centering
  \subfloat{
    \begin{tabular}{|l|c|c|}
       \hline
       & \repext & \sunflower\\
       \hline
       \comp &~$O((n+m)\Delta)$~\cite{klavik2012extending}&~$O(nm)$~\cite{jampani2012simultaneous}\\
       \hline
       \perm &~$O(n^3)$~\cite{klavik2012extending} &~$O(n^3)$~\cite{jampani2012simultaneous}\\
       \hline
       \cperm & open & open\\
       \hline
    \end{tabular}
  }
  \vspace*{4mm}
  \hfill
  \subfloat{
    \begin{tabular}{|l|c|c|}
       \hline
       & \repext & \sunflower \\
       \hline
       \comp &~$O(n+m)$ &~$O(n + m)$\\
       \hline
       \perm &~$O(n+m)$ &~$O(n + m)$\\
       \hline
       \cperm &~$O(n+m)$ &~$O(n^2)$\\
       \hline
    \end{tabular}
  }
  \caption{Known runtimes on the left and new runtimes on the right. For
    $\sunflower$ we set~$n=\sum_{i=1}^r |V(G_i)|$ and~$m=\sum_{i=1}^r
    |E(G_i)|$.  We use $\repext(\comp)$ and $\sunflower(\comp)$ to refer to
    $\orext$ and $\simor$, respectively, in a slight abuse of
  notation.}
  \label{fig: runtimes}
\end{table}

\noindent\textbf{Contribution and Outline.} In Section~\ref{sec:modDec}, we introduce modular
decompositions which can be used to describe certain subsets of the set of all transitive
orientations of a graph, e.g., those that extend a given partial representation.
Based on this, we give a simple linear-time algorithm for \orext in Section~\ref{sec:ParComp}.
Afterwards, in Section~\ref{sec: simComp}, we develop an algorithm for
intersecting subsets of transitive orientations represented by 
modular decompositions and use this to give a linear-time algorithm for \simor.
In Section~\ref{sec:perm} we give linear-time algorithms for
$\repext(\perm)$ and $\sunflower(\perm)$, improving over the
$O(n^3)$-algorithms of Klavík et al.\ and Jampani and Lubiw, respectively.
We also give the first efficient algorithms for  $\repext(\cperm)$ and $\sunflower(\cperm)$
in Section~\ref{sec:cperm}.
Table~\ref{fig: runtimes} gives an overview of the state of the art and our results.
In Section~\ref{sec:hard} we show that the simultaneous orientation 
problem and the simultaneous representation problem for permutation 
graphs are both \texttt{NP}-complete in the non-sunflower case.

\section{Modular Decompositions}
\label{sec:modDec}

Let~$G = (V, E)$ be an undirected graph.
We write $G[U]$ for the subgraph induced by a vertex set~$U\subseteq V$.  
For a rooted tree~$T$ and a node $\mu$ of $T$, we write~$T[\mu]$
for the subtree of $T$ with root~$\mu$ and~$L(\mu)$ for the leaf-set
of~$T[\mu]$.

A \emph{module} of~$G$ is a non-empty set of vertices~$M \subseteq V$ such that every
vertex~$u \in V\backslash M$ is either adjacent to all vertices in~$M$ or to
none of them.
The singleton subsets and~$V$ itself are called the \emph{trivial modules}.
A module~$M \subsetneq V$ is \emph{maximal}, if there exists no module~$M'$ such that~$M \subsetneq M' \subsetneq V$.
If~$G$ has at least three vertices and no non-trivial modules, then it is
called \emph{prime}. 
We call a rooted tree~$T$ with root~$\rho$ and~$L(\rho) = V$ a \emph{(general) modular
decomposition} for~$G$ if for every node~$\mu$ of~$T$ the set~$L(\mu)$ is a
module; see Figure~\ref{fig: mdecomp}.
\begin{figure}

    {\includegraphics{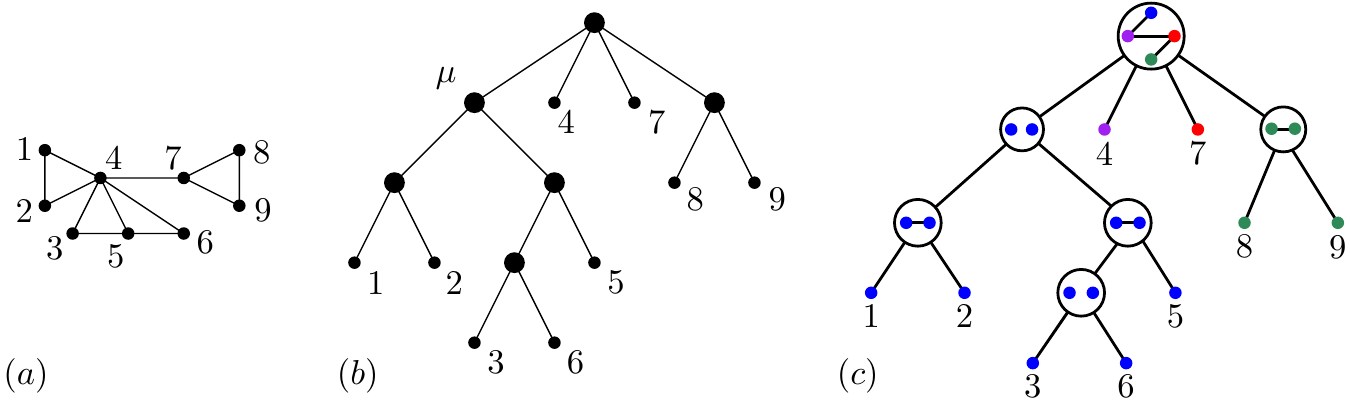}}
  \caption{(a) A graph $G$. (b) The canonical modular decomposition of~$G$ with~$L(\mu) = \{1, 2, 3, 5, 6\}$.
  (c) The quotient graphs corresponding to the nodes in the canonical modular decomposition.}
  \label{fig: mdecomp}
\end{figure}
Observe that for any two nodes~$\mu_1, \mu_2\in T$ such that 
neither of them is an ancestor of the other,~$G$ contains either all edges with one endpoint in~$L(\mu_1)$ and 
one endpoint in~$L(\mu_2)$ or none of them.
For two vertices~$u, v \in V$ we denote the lowest common ancestor of their
corresponding leaves in~$T$ by~$\mathrm{lca}_T(u, v)$. For a set of leaves $L$,
we denote the lowest common ancestor by~$\mathrm{lca}_T(L)$.

With each inner node~$\mu$ of~$T$ we associate a \emph{quotient graph}~$G[\mu]$ that is obtained from~$G[L(\mu)]$ by contracting
$L(\nu)$ into a single vertex for each child~$\nu$ of~$\mu$; see Figure~\ref{fig: mdecomp}.
In the rest of this paper we identify the vertices of~$G[\mu]$ with the
corresponding children of~$\mu$.
Every edge~$uv \in E$ is \emph{represented} by exactly one edge~$\repT{uv}$ in
one of the quotient graphs of~$T$, namely in the quotient graph~$G[\mu]$ of the
lowest common ancestor~$\mu$ of~$u$ and~$v$. 
More precisely, if~$\nu$ and~$\lambda$ are the children of~$\mu$ with~$u \in
L(\nu)$ and~$v \in L(\lambda)$, then \repT{uv}~$=$~$\nu\lambda$. 
For an oriented edge~$uv$,~$\repT{uv}$ is also oriented towards 
its endpoint~$\nu$ with~$v \in L(\nu)$.
If~$T$ is clear from the context, the subscript can be omitted. 
Let~$\mu$ be a node in~$T$. For a vertex~$u \in L(\mu)$ we denote the
child~$\lambda$ of~$\mu$ with~$u \in L(\lambda)$ by~$\text{rep}_{\mu}(u)$.

A node~$\mu$ in a modular decomposition~$T$ is called \emph{empty},
\emph{complete} or \emph{prime} if the quotient graph~$G[\mu]$ is \emph{empty},
\emph{complete} or \emph{prime}, respectively.
By~$K(T)$,~$P(T)$ we denote the set of all complete and prime nodes in~$T$,
respectively.
For every graph~$G$ there exists a uniquely defined modular decomposition,
that we call \emph{the canonical modular decomposition} of~$G$, introduced by
Gallai~\cite{gallai1967transitiv}, such that each quotient graph is either
prime, complete or empty and, additionally, no two adjacent nodes are both complete or are both empty; see Figure~\ref{fig: mdecomp}. Note that in the
literature, these are referred to as modular decompositions, whereas we use
that term for general modular decompositions.
For a prime node~$\mu$ in the canonical modular decomposition of~$G$, for every child~$\nu$ of~$\mu$,~$L(\nu)$ is a maximal module in~$G[L(\mu)]$ and for every maximal module~$M$ in~$G[L(\mu)]$ there exists a child~$\nu$ of~$\mu$ with~$L(\nu) = M$.
McConnell and Spinrad showed that the canonical modular decomposition can be computed
in~$O(|V|+|E|)$ time~\cite{mcconnell1999modular}.
Let~$\mu$ be a node in a modular decomposition for~$G$.  A
\emph{$\mu$-set} $U$ is a subset of~$L(\mu)$ that contains for each child~$\lambda$
of~$\mu$ at most one leaf in~$L(\lambda)$. 
If~$U$ contains for every child~$\lambda$ of~$\mu$ a vertex in~$L(\lambda)$,
we call it \emph{maximal}.

\begin{lemma}
\label{lemma:qoutientGraphs}
  Let~$T$ be a modular decomposition of a graph~$G$.
  After a linear-time preprocessing we can assume that each node of~$T$ is
  annotated with its quotient graph.  Moreover, the following queries can be
  answered in~$O(1)$ time:
  \begin{enumerate}[1)]
  \item Given a non-root node~$\nu$ of~$T$, find the vertex of the
    quotient graph of~$\nu$'s parent that corresponds to~$\nu$.
    
  \item Given a vertex~$v$ in a quotient graph~$G[\mu]$, find the
    child of~$\mu$ that corresponds to~$v$.

  \item Given an edge~$e$ of $G$, determine~$\rep{e}$, the quotient graph that
    contains~$\rep{e}$, and which endpoint of~$\rep{e}$ corresponds to which
    endpoint of~$e$.
  \end{enumerate}
  Additionally, given a node~$\mu$ in~$T$ one can find a maximal~$\mu$-set~$U$
  in~$O(|U|)$ time.

\end{lemma}

\begin{proof}
  We focus on constructing the quotient graphs.  The
  queries can be answered by suitably storing pointers during the
  construction.
  For every node~$\mu$ in~$T$ we initiate the quotient graph with one vertex for
  each child of~$\mu$ and equip the children of~$\mu$ and their corresponding
  vertices with pointers so that queries 1) and 2) can be answered in~$O(1)$
  time.

  Next, we compute the edges of the quotient graphs. 
  The difficulty here is to find for each edge~$uv$ in a quotient graph~$G[\mu]$ the children~$\lambda_u$,~$\lambda_v$ of~$\mu$ with~$u\in
  L(\lambda_u)$ and~$v\in L(\lambda_v)$.
  For each node~$\mu$ of~$T$ we compute a list~$L_\mu$ that contains all
  edges~$uv$ of~$G$ with~$\lca_T(u,v)=\mu$. 
  Namely, we use the lowest-common-ancestor data structure for static trees of
  Harel and Tarjan~\cite{harel1984fast}, to compute~$\lca_T(u,v)$ for each
  edge~$uv$ and add~$uv$ to~$L_{\lca(u,v)}$. 
  Afterwards, we perform a bottom-up traversal of the inner nodes
  of~$T$ that maintains for each leaf~$v$ of~$T$ the \emph{processed
    root}~$r(v) = \mu$, where~$\mu$ is the highest already-processed
  node of~$T$ with~$v \in L(\mu)$.

  Initially, we set~$r(v)=v$ for all leaves of~$T$, and we mark all
  leaves as processed.  When processing a node~$\mu$, we determine the
  edges of~$G[\mu]$ as follows.  We traverse the list~$L_\mu$ and for
  each edge~$uv \in L_\mu$, we determine the children~$r(u)$
  and~$r(v)$ of~$\mu$.  From this, we determine the corresponding
  vertices of~$G[\mu]$ and add an edge with a pointer to~$uv$
  between them. This may create multi-edges. We find those by
  sorting the incidence list of each vertex by the number of the other incident
  node in linear time using radix sort~\cite{cormen2009introduction} and an
  arbitrary enumeration of~$V(G[\mu])$. We
  then replace all parallel edges between two vertices with a single edge and
  annotate it with all pointers of the merged edges.
  For each pointer from an edge~$e$ in~$G[\mu]$ to a represented
  edge~$uv$, we then annotate~$uv$ with a pointer to~$e$. 

  Afterwards, we update~$r(v)$ for all~$v \in L(\mu)$ to~$\mu$ and mark~$\mu$
  as processed.  To maintain the processed roots of all leaves, we employ a
  union-find data structure, which initially contains one singleton set for
  each leaf and when a node~$\mu$ has been processed, we equip it with a
  maximal~$\mu$-set containing an arbitrarily chosen vertex from every set
  associated with a child of~$\mu$ and afterwards unite the sets associated
  with its children.  Since the union-find tree is known in advance (it
  corresponds to~$T$), the union-find operations can be performed in
  amortized~$\mathcal O(1)$ time~\cite{gabow1985linear}.
  
  We observe that the total size of all lists~$L_\mu$ is~$O(m)$ and
  moreover~$T$ has at most~$2n-1$ nodes.  Therefore the whole
  preprocessing runs in linear time.
\end{proof}

Let~$\mu$ be a node in a modular decomposition~$T$ of~$G$ and let~$\mu_1\mu_2$
be an edge in~$G[\mu]$.
Let~$\vec{T}[G]$ denote an assignment of directions to all
edges in~$G[\mu]$ for every node~$\mu$ in~$T$. Such a~$\vec{T}[G]$ is
\emph{transitive} if it is transitive on every~$G[\mu]$.
We obtain an orientation of~$G$ from an orientation~$\vec{T}[G]$ of the  quotient graphs of~$T$ as follows. 
Every undirected edge~$uv$ in~$E$ with \rep{uv}~$= \mu_1\mu_2 \in \vec{T}[G]$,
is directed from~$u$ to~$v$. We say that~$T$ \emph{represents} an
orientation~$O$ of~$G$, if there exists an orientation~$\vec{T}[G]$ of the  quotient graphs of $T$ that gives us~$O$.
We denote the set of all transitive orientations of~$G$ represented by~$T$
by~$\TO(T)$.
We get an orientation of the quotient graphs of~$T$ from an orientation of~$G$, if for each
oriented edge~$\mu_1\mu_2$, all edges represented by~$\mu_1\mu_2$ are oriented
from~$L(\mu_1)$ to~$L(\mu_2)$.
Let~$T$ now be the canonical modular decomposition of~$G$.
Then~$T$ represents exactly the transitive orientations of~$G$~\cite{gallai1967transitiv}.
It follows that~$G$ is a comparability graph if and only if~$T$ can be oriented
transitively.

If $G$ is a comparability graph, every prime quotient graph~$G[\mu]=(V_\mu,E_\mu)$ has exactly two transitive
orientations, one the reverse of the 
other~\cite{golumbic2004algorithmic}, and with the algorithm by McConnell
and Spinrad~\cite{mcconnell1999modular} we can compute one of them
in~$O(|V_\mu|+|E_\mu|)$ time.
Hence the time to compute the canonical modular decomposition in which every
prime node is labeled with a corresponding transitive orientation
is~$O(|V|+|E|)$.

\section{Transitive Orientation Extension}
\label{sec:ParComp}

The \textit{partial orientation extension problem} for comparability graphs
\orext is to decide for a comparability graph~$G$ with a partial
orientation~$W$, i.e.~an orientation of some of its edges, whether there exists
a transitive orientation~$O$ of~$G$ with~$W \subseteq O$.  The notion of
partial orientations and extensions extends to modular decompositions. We get a partial orientation of the quotient graphs of~$T$ from~$W$ such that exactly the edges that
represent at least one edge in~$W$ are oriented and all edges in~$W$ that are
represented by the same oriented edge~$\mu_1\mu_2$, are directed
from~$L(\mu_1)$ to~$L(\mu_2)$.

\begin{lemma}
\label{lemma: 1}
 Let~$T$ be the canonical modular decomposition of a comparability graph~$G$
 and let~$W$ be a partial orientation of~$G$ that gives us a partial
 orientation~$P$ of the quotient graphs of~$T$.
 Then~$W$ extends to a transitive orientation of~$G$ if and only if~$P$
 extends to a transitive orientation of the quotient graphs of~$T$.
\end{lemma}

\begin{proof}
  Let~$O$ be a transitive orientation of~$G$ that extends~$W$.
  Let~$\mu_1\mu_2$ be an oriented edge in~$P$.
  Then~$\mu_1\mu_2$ represents an oriented edge~$uv$ in~$W$.
  Then~$uv$ is an oriented edge in~$O$ and~$\mu_1\mu_2$ is in the transitive
  orientation~$\vec{T}[G]$ of the quotient graphs of~$T$ we get from~$O$. Hence,~$\vec{T}[G]$ extends~$P$.
  
  Conversely, let~$\vec{T}[G]$ be a transitive orientation of the quotient graphs of~$T$ that extends~$P$.
  Let~$uv$ be an oriented edge in~$W$.
  Then~$uv$ is represented by an oriented edge~$\mu_1\mu_2$ in~$P$.
  Then~$\mu_1\mu_2$ is an oriented edge in~$\vec{T}[G]$ and~$uv$ is in the
  transitive orientation~$O$ of~$G$ we get from~$\vec{T}[G]$. Hence,$O$ extends~$W$.
\end{proof}

  To solve \orext efficiently we confirm that the partial orientation actually gives us a partial orientation~$P$ of the quotient graphs of
  the canonical modular decomposition~$T$. Otherwise we can reject.
  By Lemma~\ref{lemma: 1} we now just need to check for each node~$\mu$ of $T$
  whether~$P$ can be extended to $G[\mu]$.
  To this end, we use that $\mu$ is empty, complete or prime.
  Since transitive orientations of cliques are total orders and prime graphs
  have at most two transitive orientations, the existence of an extension can
  easily be decided in each~case.

\begin{theorem}
  \orext can be solved in linear time. 
\end{theorem}

\begin{proof}  Let~$W$ be the given partial orientation of a comparability graph~$G=(V,E)$. 
After the linear-time
preprocessing of Lemma~\ref{lemma:qoutientGraphs}, we can compute the partial
orientation~$P$ of the quotient graphs of~$T$ we get from~$W$ in linear time by determining~$\rep{uv}$
for every edge~$uv \in W$. If~$P$ does not exist, then there is an 
edge~$\mu_1\mu_2$ in a quotient graph that represents an edge~$e_1\in W$
oriented from~$L(\mu_1)$ to~$L(\mu_2)$ and an edge~$e_2\in W$ oriented
from~$L(\mu_2)$ to~$L(\mu_1)$.  Then~$W$ can not be extended to a transitive
orientation of~$G$ since in any orientation represented by~$T$ the
edges~$e_1$,\,$e_2$ are both oriented in the same direction between~$L(\mu_1)$
and~$L(\mu_2)$. Hence, we can reject in this case.

Otherwise, to solve \orext for~$G$, it suffices to solve \orext
for every quotient graph in the canonical modular decomposition~$T$ of~$G$ with the
partial orientation from~$P$ by Lemma~\ref{lemma: 1}.
Let~$\mu$ be a node in~$T$. We distinguish cases based on the type of~$\mu$.
If~$\mu$ is empty, nothing needs to be done.
If~$\mu$ is complete, the problem of extending the partial orientation
of~$G[\mu]$ is equivalent to the problem of finding a total order of the nodes
of~$G[\mu]$ that respects~$P$. This can be done via topological sorting in linear
time.
If~$\mu$ is prime,~$G[\mu]$ has exactly two transitive orientations, where one
is the reverse of the other.
Therefore we check in linear time whether one of these orientations of~$G[\mu]$
 is an extension of the partial orientation of~$G[\mu]$. Otherwise, no
 transitive extension exists. 

Since we can compute~$T$ in~$O(|V|+|E|)$ time, in total we can decide whether
the partial orientation~$W$ is extendible in the same time.
we get a corresponding transitive orientation of~$G$ by the extension
of~$P$ and can also be computed in the same time.
\end{proof}

\section{Sunflower Orientations}
\label{sec: simComp}

The idea to solve \simor is to obtain for each
input graph~$G_i$ a restricted modular decomposition of the shared
graph~$H$ that represents exactly those transitive orientations of~$H$ that can
be extended to~$G_i$. The restricted modular decompositions can be expressed by
constraints for the canonical modular decomposition of~$H$. These constraints are then
combined to represent all transitive orientations of~$H$ that can be extended
to each input graph~$G_i$. With this the solution is straightforward.

Let~$H$ be a comparability graph. Then we define a 
\emph{restricted modular decomposition}~$(T,D)$ of~$H$ to be a tuple where~$T$
is a modular decomposition of~$H$ where every node is labeled as complete,
empty or prime, such that for every node~$\mu$ labeled as complete or
empty,~$H[\mu]$ is complete or empty, respectively, and~$D$ is a function that
assigns to each prime labeled node~$\mu$ a transitive orientation~$D_\mu$,
called \emph{default orientation}.
In the following, when referring to the type of a node~$\mu$ in a restricted
modular decomposition, we mean the type that~$\mu$ is labeled with.
A transitive orientation of~$(T,D)$, is a transitive orientation of the quotient graphs of~$T$ where
every prime node~$\mu$ has orientation~$D_\mu$ or its reversal~$D_\mu^{-1}$.
Let~$\TO(T, D)$ denote the set of transitive orientations of~$H$ we get from
transitive orientations of~$(T,D)$. We say~$(T,D)$ \emph{represents} these
transitive orientations.  Note that for the canonical modular decomposition~$B$ of~$H$ we
have~$\TO(T, D) \subseteq \TO(B)$.  

Let~$G$ be a comparability graph with an induced subgraph~$H$.  A modular
decomposition~$T$ of~$G$ gives us a restricted modular decomposition~$(T|_H,D)$
of~$H$ as follows; see Figure~\ref{fig:reducedMD}.
\begin{figure}[!t]
  \centering
  {\includegraphics{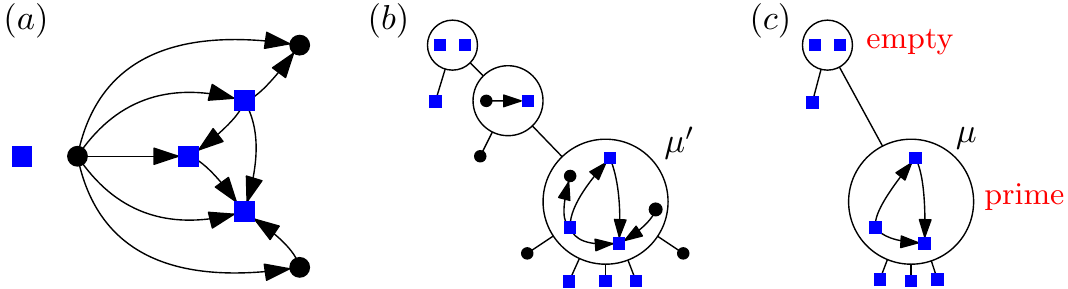}}
  \caption{(a) A graph~$G$ with an induced subgraph~$H$ (blue square vertices).
    (b) A modular decomposition~$T$ of~$G$ with a transitive orientation. (c)
    The restricted modular decomposition~$(T|_H, D)$ of~$H$ derived from the
  transitive orientation in (b). Note that $H[\mu]$ is a clique but $\mu$ is
  labeled prime.}
  \label{fig:reducedMD}
\end{figure}
We obtain~$T|_H$ from~$T$ by (i) removing all leaves that do not
correspond to a vertex of~$H$ and then (ii) iteratively contracting all inner
nodes of degree at most~2.
With a bottom-up traversal we can compute~$T|_H$ in time linear in the size
of~$T$.
A node~$\mu$ in~$T|_H$ \emph{stems from}~$\lca_T(L(\mu))$.
Every node~$\mu \in T|_H$ that stems from a prime node~$\mu' \in T$ we label as
prime and set~$D_\mu$ to a transitive orientation of~$G[\mu']$ restricted
to the edges of~$H$. 
The remaining nodes are labeled according to the type of their quotient graph.
Note that~$H[\mu]$ is isomorphic to an induced subgraph of~$G[\mu']$.

\begin{lemma}
\label{lemma: mdS}
  Let~$T$ be a modular decomposition of a graph~$G$ with an induced subgraph~$H$.
  Then~$\TO(T|_H, D)$ is the set of orientations of~$H$ extendable to
  transitive orientations of~$G$. 
\end{lemma}
 
\begin{proof}
 Let~$O_H$ be a transitive orientation of~$H$
 that can be extended to a transitive orientation~$O_G$ of~$G$.
 Then~$O_G$ gives us a transitive orientation~$\vec{T}[G]$ of the quotient graphs of~$T$ that in turn gives us a transitive
 orientation~$\vec T|_H$ of~$(T|_H,D)$.
 Since~$O_G$ is an extension of~$O_H$ and~$\vec{T}[G]$ contains~$\vec{T}_H[H]$,~$O_H$ is the transitive orientation of~$H$ we get from~$\vec{T}_H[H]$. Hence~$(T|_H,D)$ represents~$O_H$.
  
  Conversely, let~$\vec{T}_H[H]$ be a transitive orientation of~$(T_H,D)$. 
  For any node~$\mu$ in~$T|_H$ let~$\mu' = \mathrm{lca}_T(L(\mu))$.
  Recall that~$H[\mu]$ is isomorphic to an induced 
  subgraph of~$G[\mu']$.
  We already know that~$G[\mu']$ is either empty, complete or prime.
  If~$G[\mu']$ is empty,~$H[\mu]$ and the corresponding transitive orientation
  are also empty. 
  If~$G[\mu']$ is complete, then~$H[\mu]$ is also complete and any transitive
  orientation of~$H[\mu]$ can be extended to a transitive orientation
  of~$G[\mu']$.
  If~$G[\mu']$ is prime, by construction~$H[\mu]$ is also labeled as prime with
  a default orientation~$D_{\mu}$ given by a transitive orientation~$D_{\mu'}$
  of~$G[\mu']$.
Hence~$\vec{T}_H[H]$ either contains~$D_{\mu'} \cap E(H[\mu]) = D_{\mu}$
or~$D_{\mu'}^{-1} \cap E(H[\mu]) = D_{\mu}^{-1}$.
  Thus~$\vec{T}_H[H]$ can be mapped to~$T$ and be extended to a transitive orientation~$\vec{T}[G]$ of the quotient graphs of~$T$.
  Then~$\vec{T}[G]$ gives us a transitive orientation~$O_G$ of~$G$ since~$T$ represents exactly the transitive orientations of~$G$.
  Let~$O_H$ be~$O_G$ restricted to~$H$. Then by construction~$O_H$ equals the
  orientation of~$H$ we get from~$\vec{T}_H[H]$. Thus~$\vec{T}_H[H]$ gives us a
  transitive orientation of~$H$.
\end{proof}

Consider the canonical modular decompositions $T^1,\dots,T^r$ of the input graphs $G_1,\dots,G_r$,
and let $(T^1|_H,D_1),\dots, (T^r|_H, D_r)$ be the corresponding restricted modular decompositions.
Then we are interested in the intersection $\TO(G_1,\dots,G_r)=\bigcap_{i=1}^r\TO(T^i|_H,D_i)$
since it contains all transitive orientations of $H$ that can be extended to
all input graphs. However, the trees $T^1|_H,\dots,T^r|_H$ have different
shapes, which makes it difficult to compute a representation of
$\TO(G_1,\dots,G_r)$ directly. Instead, we describe $\TO(T^1|_H,D_1),\dots, \TO(T^r|_H, D_r)$ (whose
intersection is simple to compute)
with constraints on the canonical modular decomposition~$B$ of~$H$.

Let~$(T,D)$ be a restricted modular decomposition of~$H$ and let~$B$ be the
canonical modular decomposition of~$H$.
We collect constraints on the orientations of individual nodes $\mu$ of~$B$
that are imposed by $\TO(T,D)$. Afterwards we show that the established
constraints are sufficient to describe $\TO(T,D)$.
If~$\mu$ is empty, then~$H[\mu]$ is empty and has a unique transitive
orientation, which requires no constraints. The other types of~$\mu$ are
discussed in the following two sections.

\subsection{Constraints for Prime Nodes}

In this section we observe that prime nodes in~$B$ correspond to prime nodes
in~$(T,D)$ and that the dependencies between their orientations can be
described with 2-SAT formulas.
Recall that the transitive orientation of a prime comparability graph is unique up to reversal.
We consider each prime node~$\mu\in B$ equipped
with a transitive orientation~$D_\mu$, also called a \emph{default orientation}.
All default orientations for~$B$ can be computed in linear time~\cite{mcconnell1999modular}.

\begin{lemma}
  \label{lemma: prim}
  \label{lemma: prime2}
  Let~$\mu\in B$ be prime. Then~$\mu'=\lca_T(L(\mu))$ is prime, 
  every~$\mu$-set is a~$\mu'$-set, and for any edge~$uv$ with~$\repB{uv} \in
  H[\mu]$ we have~$\repT{uv}\in H[\mu']$.
\end{lemma}
\begin{proof}
  We first show that every~$\mu$-set is also a~$\mu'$-set.
  Assume there is a~$\mu$-set~$U$ that is not a~$\mu'$-set.
  Since~$U\subseteq L(\mu)\subseteq L(\mu')$, there exist two 
  vertices~$u\ne v \in U$ with~$\lambda =\text{rep}_{\mu'}(u) =
  \text{rep}_{\mu'}(v)$.
  From~$L(\lambda)$ being a module of~$H[L(\mu')]$ and~$L(\mu) \subseteq
  L(\mu')$ it follows that~$X=L(\lambda) \cap L(\mu)$ is a module
  of~$H[L(\mu)]$. By the definition of~$\mu'$, we have
 ~$X\subsetneq L(\mu)$. %
  Since~$\mu$ is prime, there is a child~$\nu$ of~$\mu$ with~$u,v\in
  X\subseteq L(\nu)$ contradicting~$U$ being a~$\mu$-set. 
  
  As a direct consequence of~$\mu$-sets being~$\mu'$-sets, we also have for any
  edge~$uv$ with~$\repB{uv} \in H[\mu]$ that~$\repT{uv}\in H[\mu']$ since~$\{u,v\}$ is a~$\mu$-set. 
  Now let~$U$ be a maximal~$\mu$-set.
  Then~$H[U]$ is isomorphic to~$H[\mu]$ and thus
  prime. Since~$U$ is also a~$\mu'$-set and the subgraph
  of~$H[\mu']$ induced by the vertices representing~$U$ is isomorphic
  to~$H[U]$,~$\mu'$ is neither empty, nor complete and thus prime.
\end{proof}

For a modular decomposition~$T'$ of~$H$ with a node~$\lambda$ and~$O\in \TO(T')$
let~$O_{\downarrow \lambda}$ denote the orientation of the quotient
graph~$H[\lambda]$ we get from~$O$.
Note that for a transitive orientation~$D_\lambda$ of~$H[\lambda]$ and
a~$\lambda$-set~$U$ each orientation~$O\in\TO(T')$
with~$O_{\downarrow\lambda}= D_\lambda$ gives us the same orientation on~$H[U]$.
We say that~$D_\lambda$ \emph{induces} this orientation on~$H[U]$.

Let~$\mu\in B$ be prime, let~$\mu'=\lca_T(L(\mu))$ and let~$U$ be a
maximal~$\mu$-set (and by Lemma~\ref{lemma: prim} a~$\mu'$-set).
We set~$D^{\delta}_{\mu'}=D_{\mu'}$ if~$D_\mu$,~$D_{\mu'}$ induce the same
transitive orientation on the prime graph~$H[U]$ and we set~$D^{\delta}_{\mu'}=D^{-1}_{\mu'}$
 if the induced orientations are the reversal of each other.
Note that~$D^{\delta}_{\mu'}$ does not depend on the choice of~$U$.
From the definition of~$D^{\delta}_{\mu'}$ and the observation
that~$O_{\downarrow \mu}, O_{\downarrow \mu'}$ are both determined by~$O$ restricted to~$H[U]$ we directly get the following lemma.

\begin{lemma}
  \label{lem:primChar}
  For~$O\in\TO(T,D)$ we have~$O_{\downarrow\mu} 
  =D_\mu \Leftrightarrow O_{\downarrow\mu'} = D^{\delta}_{\mu'}$ 
\end{lemma}

We express the choice of a transitive orientation for a prime node~$\mu$ by a
Boolean variable~$x_{\mu}$ that is~$\mathtt{true}$ for the default orientation
and~$\mathtt{false}$ for the reversed orientation.

According to Lemma~\ref{lem:primChar} we set~$\psi_{\mu}$ to be $x_\mu \leftrightarrow x_{\mu'}$ 
if~$D_{\mu'}^{\delta}= D_{\mu'}$ and~$x_\mu \not\leftrightarrow x_{\mu'}$ if~$D_{\mu'}^{\delta}= D^{-1}_{\mu'}$.
Note that for a prime node~$\mu' \in T$ there may exist more than one
prime node~$\mu$ in~$B$ such that~$\mu' = \mathrm{lca}_T(L(\mu))$, and we may
hence have multiple prime nodes that are synchronized by these constraints.  We
describe these dependencies with the formula~$\psi = \bigwedge_{\mu \in
P(B)} \psi_\mu$.  With the above meaning of
variables, any choice of orientations for the prime nodes of~$B$ that can be
induced by~$T$ necessarily satisfies~$\psi$.
With Lemma~\ref{lemma:qoutientGraphs} we can compute~$\psi$
efficiently.

\begin{lemma}
\label{lemma: phiprime}
We can compute~$\psi$ in~$O(|V(H)|+|E(H)|)$ time.
\end{lemma}

\begin{proof}
 Let~$\mu$ be a prime node in~$B$ and let~$\mu'= \mathrm{lca}_T(L(\mu))$.
 By Lemma~\ref{lemma:qoutientGraphs} we can compute a
 maximal~$\mu$-set~$U$ in constant time after a linear-time preprocessing.
 By Lemma~\ref{lemma: prim} we have for every edge~$uv$
 with~$\repNode{B}{uv}\in H[\mu]$ that~$\repNode{T}{uv}\in H[\mu']$.
 Hence we can find~$\mu'$ in~$T$ by determining~$\repNode{T}{uv}$ for an
 arbitrary edge~$uv$ with ~$u, v \in U$, which by
 Lemma~\ref{lemma:qoutientGraphs} takes constant time. For an arbitrary
 oriented edge~$e \in O$ we check in constant time whether~$e \in D_\mu$ or~$e
 \in D_\mu^{-1}$ and add the clause~$x_\mu \leftrightarrow x_{\mu'}$ or~$x_\mu
 \not\leftrightarrow x_{\mu'}$, respectively. Doing this for every prime node in
 total takes time linear in the size of~$T$.
\end{proof}

\subsection{Constraints for Complete Nodes} 
Next we consider the case where~$\mu$ is complete.
The edges represented in~$H[\mu]$  
may be represented by edges in more than one quotient graph in~$T$, each of
which can be complete or prime. Depending on the type of the involved quotient
graphs in~$T$ we get new constraints for the orientation of~$H[\mu]$.

Note that choosing a transitive orientation of~$H[\mu]$ is
equivalent to choosing a linear order of the vertices of~$H[\mu]$.
As we will see, each node~$\nu$ of~$T$ that represents an edge
of~$H[\mu]$ imposes a consecutivity constraint on a subset of the
vertices of~$H[\mu]$.
Therefore, the possible orders can be represented by a \emph{PQ-tree} that
allows us to represent all permissible permutations of the elements of a
set~$U$ in which certain subsets~$S \subseteq U$ appear consecutively.

\begin{figure}[!t]
  \centering{
    \includegraphics{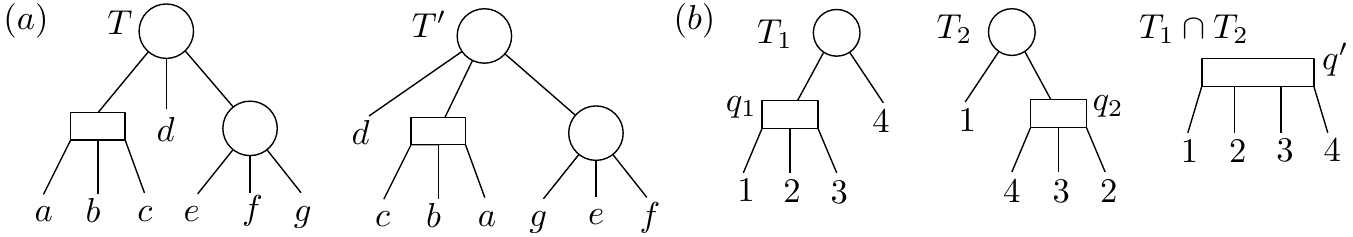}
  }
  \caption{(a)~Two equivalent PQ-trees~$T$,~$T'$ with~$fr(T) = abcdefg$
  and~$fr(T')=dcbagef$. (b)~The Q-node~$q'$ in~$T_1 \cap T_2$ contains~$q_1$
forwards and~$q_2$ backwards.}
  \label{fig: pq}
\end{figure}

PQ-trees were first introduced by Booth and Lueker~\cite{booth1975pq, booth1976testing}.
A PQ-tree~$T$ over a finite set~$U$ is a rooted tree whose leaves are the
elements of~$U$ and whose internal nodes are labeled as~$P$- or~$Q$-nodes. 
A~$P$-node is depicted as a circle, a~$Q$-node as a rectangle; see Figure~\ref{fig: pq}.
Two PQ-trees~$T$ and~$T'$ are \emph{equivalent}%
, if~$T$ can be transformed into~$T'$ by arbitrarily
permuting the children of arbitrarily many~$P$-nodes and reversing the order of
arbitrarily many~$Q$-nodes; see Figure~\ref{fig: pq}a. A transformation that transforms~$T$ into an 
equivalent tree~$T'$ is called an \emph{equivalence transformation}.
The \emph{frontier}~$\fr(T)$ of a PQ-tree~$T$ is the order of its leaves from left to right.
The tree~$T$ \emph{represents} the frontiers of all equivalent~$PQ$-trees.
The PQ-tree that does not have any nodes is called the \emph{null tree}.

  Let~$T_1$,~$T_2$ be two PQ-trees over a set~$U$.
  Their \emph{intersection}~$T=T_1\cap T_2$ is a PQ-tree that represents exactly
  the linear orders of~$U$ represented by both~$T_1$ and~$T_2$.
  It can be computed in~$O(|U|)$ time~\cite{booth1975pq}.
  For every Q-node~$q$ in~$T_1$ node~$q'=\lca_T(L(q))$ is also a Q-node.
  We say that~$q'$ \emph{contains}~$q$ \emph{forwards}, if~$T_1[q]$ can be
  transformed by an equivalence transformation that does not reverse~$q$ into a
  PQ-tree~$T'$ such that~$\fr(T[q])$ contains~$\fr(T[q'])$; see
  Figure~\ref{fig: pq}b. 
 Else~$q'$ \emph{contains}~$q$ \emph{backwards}.
 Similarly every Q-node in~$T_2$ is contained in exactly one Q-node in~$T$
 (either forwards or backwards).
  Haeupler et al.~\cite{haeupler2013testing} showed that one can modify Booth's
  algorithm such that given two PQ-trees~$T_1, T_2$ it not only outputs ~$T_1
  \cap T_2$ but also for every Q-node in~$T_1, T_2$ which Q-node in~$T$
  contains it and in which direction.
  
  \begin{lemma}
  \label{lemma:pqIntersection}
    Let~$T_1, \dots, T_k$ be PQ-trees over a set~$U$.
    Then we can compute their intersection~$T = \bigcap_{i = 1}^k T_i$  and
    determine for every Q-node~$q$ in~$T_1, \dots, T_k$
    the Q-node in~$T$ that contains~$q$ and in which direction in~$O(k \cdot |U|)$ time.
  \end{lemma}
 
  \begin{proof}
    Let~$S_1 = T_1$ and for every~$j \in \{2, \dots, k\}$ let ~$S_j : = S_{j-1} \cap T_i$.
    To compute~$T = S_k$ we stepwise
    compute~$S_j = S_{j-1} \cap T_j$ for every~$j$. 
    During the computation of~$T$ we construct a DAG~$D$ whose vertices are the
    Q-nodes of~$T_1, \dots, T_k$ and~$S_2, \dots, S_k$. 
    Initially~$D$ contains no edges.
  	For every~$j$ we add a directed edge from every Q-node~$q$
  	in~$T_j$ and~$S_{j-1}$ to the Q-node~$q'$ in~$S_j$ that contains~$q$.
    We label the edge with 1 if~$q'$ contains~$q$ forwards, and with -1 otherwise. 
    By the result of Haeupler et al.~this can be done in~$O(k\cdot |U|)$ time~\cite{haeupler2013testing}.
    Note that by construction every vertex has at most one outgoing edge and
    for every Q-node~$q$ in~$T_1,\dots, T_k$ there is a unique path to the
    Q-node~$q'$ in~$S_k$ that contains it. The product of the edge labels along
    this path is 1 and -1 if~$q'$ contains~$q$ forwards and backwards,
    respectively.
    
    To determine for every Q-node~$q$ in~$T_1,\dots, T_k$ which Q-node in~$T$
    contains it and in which direction, we start at the sinks in~$D$ and
    backward propagate for every vertex in~$D$ the information which unique
    sink~$q'$ can be reached from it and what is the product of edge labels
    along the path to~$q'$.
This needs~$O(k \cdot |U|)$ time since the~$D$ has~$O(k \cdot |U|)$ vertices
and edges. 
  \end{proof}

Let~$\mu'=\mathrm{lca}_T(L(\mu))$ for the rest of this section 
 \begin{figure}
  \centering
  \includegraphics{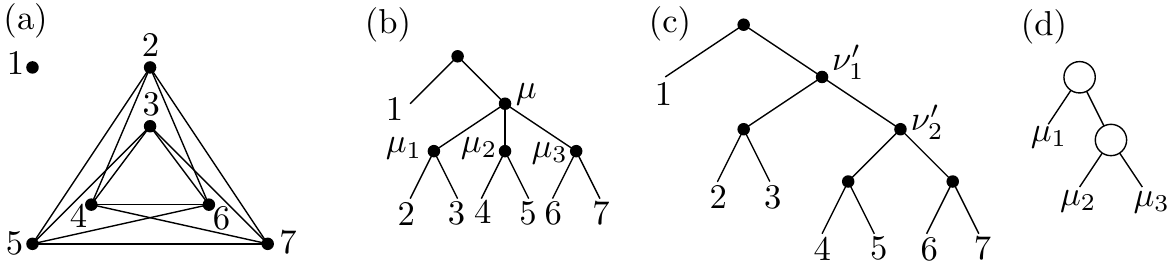}
  \caption{(a) A graph~$H$. (b) The canonical modular decomposition~$B$ of~$H$. (c) A restricted modular decompositions~$T$ of~$H$.
  (d) The PQ-tree~$S_\mu$.  The active nodes of~$T$ with regard to the~$\mu$-set~$\{2, 4, 6\}$ are~$\nu_1'$,~$\nu_2'$, 2, 4 and 6.}
  \label{fig: PQ}
\end{figure}
and let~$U \subseteq V$ be a maximal~$\mu$-set.  We call a node of~$T$
\emph{active} if it is either a leaf in~$U$ or if at least two of
its subtrees contain leaves in~$U$.  Denote by~$A$ the set of
active nodes in~$T$ and observe that~$A$ can be turned into a tree~$S_\mu$ by
connecting each node~$\nu' \in A \setminus \{\mu'\}$ to its lowest
active ancestor; see Figure~\ref{fig: PQ}.
Let now~$\vec B[H]$ be an orientation of the quotient graphs of ~$B$ induced by an orientation~$\vec
T[G] \in\TO(T, D)$ and consider a node~$\nu'\ne \mu'$ of~$S_\mu$.
Let~$X = U\cap L(\nu')$ and let~$Y = U \setminus L(\nu')$.  Since~$U$ is
a~$\mu$-set of a complete node, any pair of vertices in~$X \times Y$ is
adjacent.
Moreover, for each~$y \in Y$, the edges from~$y$ to~$X$ are all
oriented towards~$X$, or they are all oriented towards~$y$, since
every node of~$T$ that determines the orientation of such an edge
contains all vertices of~$X$ in a single child. This implies that in the
order of the~$\mu$-set given by the order of~$H[\mu]$, the
set~$L(\nu')$ is consecutive.  Moreover, if~$\nu'$ is prime, its
default orientation~$D_{\nu'}$ induces a total order on the active
children of~$\nu'$ that is fixed up to reversal.  Hence we turn~$S_\mu$ into a
PQ-tree by first turning all complete nodes into P-nodes and all prime nodes
into Q-nodes with the children ordered according to the linear order determined
by the default orientation which we call the \emph{initial} order of the
Q-node. Finally, we replace each leaf~$v \in U$ by the corresponding
vertex~$\mathrm{rep}_\mu(v)$ of~$H[\mu]$; see Figure~\ref{fig: PQ}.  As argued
above, the linear order of~$H[\mu]$ is necessarily represented by~$S_\mu$.

We show that tree~$S_\mu$ is independent from the choice of the maximal~$\mu$-set~$U$.
We use that any node of~$T$ has a laminar relation to the children of~$\mu$. 

\begin{lemma}
\label{lemma: rotBlau}
  For any child~$\mu_1$ of~$\mu$ and any node~$\kappa'$ of~$V(T)$ we
  have~$L(\mu_1)\subseteq L(\kappa')\cap L(\mu)$ or~$L(\kappa')\cap
  L(\mu)\subseteq L(\mu_1)$. 
\end{lemma}

  \begin{figure}
  \centering
  {\includegraphics{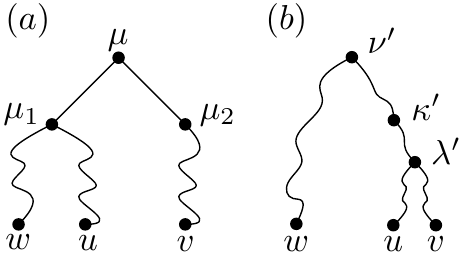}}
  \caption{(a) A subtree of~$B$ with root~$\mu$. (b) A subtree of~$T$ with
  root~$\nu'$.}
  \label{fig: rotBlau}
\end{figure}
\begin{proof}
  If~$L(\kappa')$ contains only leaves of at most one child~$\mu_1$ of~$\mu$,
  we have~$ L(\kappa')\cap L(\mu)\subseteq L(\mu_1)$
  and~$L(\mu_2)\subseteq L(\kappa')\cap L(\mu)$ for each other
  child~$\mu_2$ of~$\mu$.
  Otherwise, let~$\mu_1$,\,$\mu_2$ be children of~$\mu$ with leaves~$u
  \in L(\kappa') \cap L(\mu_1)$,~$v \in  L(\kappa') \cap L(\mu_2)$ and
  let~$\lambda'= \mathrm{lca}_T(u, v)$; see Figure~\ref{fig: rotBlau}. 
  Note that~$\lambda'$ is a descendant of~$\kappa'$ or~$\lambda' = \kappa'$.
  Assume that there exists a vertex~\mbox{$w \in L(\mu_1)\setminus L(\kappa')$} and
  let~$\nu' = \mathrm{lca}_T(u, w)$. Then~$\kappa'$ is a descendant of~$\nu'$
  since~$u\in L(\nu')\cap L(\kappa')$ and~$w\in L(\nu')\setminus L(\kappa')$.
  Note that~$\repS{vw} = \repS{uw}\in H[\nu']$ and~$\repS{uv}\in H[\lambda']$.
  Hence,~$T$ represents a transitive orientation of~$H$ with~$uv$ and~$vw$.
  This contradicts~$B$ representing all transitive orientations of~$H$ since we
  have~$\repB{uv}=\repB{wv}$.
  It follows that~$L(\mu_1)\setminus L(\kappa')=\emptyset$ and
  analogously~$L(\mu_2)\setminus L(\kappa')=\emptyset$. This concludes the
  proof.
\end{proof}

\begin{lemma}
  Let~$S_\mu^1$,\,$S_\mu^2$ be the PQ-trees for two
  maximal~$\mu$-sets~$U_1$,\,$U_2$.  Then~$S_\mu^1 = S_\mu^2$.
\end{lemma}

\begin{proof}
  Let~$\nu'$ be a non-leaf node in~$T$ that is active with regard to~$V_1'$.
  Then there exist two vertices~$u_1, v_1 \in V_1'$ and two distinct
  children~$\nu_1'$,~$\nu_2'$ of~$\nu'$ such that~$u_1 \in L(\nu_1')$ and~$v_1
  \in L(\nu_2')$.
  Let~$\mu_1 = \repMu{u_1}$,~$\mu_2 = \repMu{v_1}$ and let~$u_2$,\,$v_2$ be the
  vertices in~$V_2'$ 
  with~$\repMu{u_2} = \mu_1$,~$\repMu{v_2} = \mu_2$.
  We prove the following:
  \begin{enumerate}[(i)]
    \item \label{moduloMuSet:itm1} ~$u_2,v_2\in L(\nu')$ 
    \item \label{moduloMuSet:itm2}~$\repNode{\nu'}{u_2}\ne \repNode{\nu'}{v_2}$
    \item \label{moduloMuSet:itm3}~$\repNode{\nu'}{u_1v_1}\in D_{\nu'}
      \Leftrightarrow \repNode{\nu'}{u_2v_2}\in D_{\nu'}$.  
  \end{enumerate}
  For Statement~\ref{moduloMuSet:itm1} we get~$u_2\in L(\mu_1)\subseteq L(\nu')\cap
  L(\mu)$ by Lemma~\ref{lemma: rotBlau} since~$v_1\not\in L(\mu_1)$ and
  similarly~$v_2\in L(\nu')\cap L(\mu)$. 

  For Statement~\ref{moduloMuSet:itm2} assume~$\nu''=\repNode{\nu'}{u_2} = \repNode{\nu'}{v_2}$.
  Then we have~$u_2\in L(\nu'')$ and~$u_2\not\in L(\mu_1)$. By
  Lemma~\ref{lemma: rotBlau} we get~$u_1\in L(\mu_1)\subseteq L(\nu'')\cap
  L(\mu)$, i.e.,~$\nu''=\nu'_1$ and similarly~$\nu''=\nu'_2$ in contradiction
  to~$\nu'_1\ne\nu'_2$.

  For Statement~\ref{moduloMuSet:itm3} note that~$u_1v_1$ and~$u_2v_2$ are
  indeed represented in~$\nu'$ by Statement~\ref{moduloMuSet:itm2} and they are
  both represented by~$\mu_1\mu_2$ in~$H[\mu]$.  

  It follows that the active inner nodes are the same for~$U_1$,\,$U_2$ and
  after replacing the leaves with the children of~$\mu$ we obtain the same trees.
  Statement~\ref{moduloMuSet:itm3} then provides that the ordering of the
  children of~$Q$-nodes is also the same and~$S_\mu^1$,\,$S_\mu^2$ are indeed
  the same PQ-tree.
\end{proof}

By construction, each~Q-node~$q$ of~$S_\mu$ stems from a prime
node~$\nu'$ in~$T$, and the orientation of~$T[\nu']$ determines the
orientation of~$q$, namely~$q$ is reversed if and only if~$T[\nu']$ is
oriented as~$D_{\nu'}^{-1}$.  Since a single prime node~$\nu'$ of~$T$
may give rise to Q-nodes in several PQ-trees~$S_{\mu}$, we need to
ensure that the orientations of these Q-nodes are either all consistent
with the default orientation of~$T[\nu']$ or they are all consistent
with its reversal. To model this, we introduce a Boolean
variable~$x_q$ for each Q-node~$q$ in one of the PQ-trees with the interpretation that~$x_q = \mathtt{true}$ if and only if~$q$ has its initial order. We
require~$x_q$ to be equal to the variable that orients the 
prime node corresponding to~$q$.  More precisely, for every prime node~$\nu'$ in~$T[\mu']$ that gives rise to~$q$
we add the constraint~$(x_{\nu'} \leftrightarrow x_q)$ to~$\chi_{\mu}$, where the variable~$x_{\nu'}$ is the variable that
encodes the orientation of the prime node~$\nu'$.  We
construct a Boolean formula by setting~$\chi = \bigwedge_{\mu \in K(B)}
\chi_{\mu}$.

\begin{lemma}
  \label{lemma: linear-time PQ} 
  We can compute all PQ-trees~$S_\mu$ and the formula~$\chi$ in~$O(n + m)$ time.  
\end{lemma}
\begin{proof}
  As a preprocessing we run a DFS on~$T$ starting at the root and
  store for every node~$\nu$ its discovery-time~$\nu.d$, i.e., the
  timestamp when~$\nu$ is first discovered, and its
  finish-time~$\nu.f$, i.e., the timestamp after all its neighbors
  have been examined.  We also employ the preprocessing from
  Lemma~\ref{lemma:qoutientGraphs}.
  We construct all PQ-trees and~$\chi$ with the following steps.
  \begin{enumerate}
  \item Take a maximal~$\mu$-set~$U_\mu$ for every~$\mu \in K(B)$.
  \item For every~$\mu \in K(B)$ compute the set of active nodes and for every
    active node compute its parent in~$S_\mu$.
    
  \item For every~$\mu \in K(B)$ determine for each inner node of~$S_\mu$
    whether it is a P- or a Q-node. If it is a Q-node, determine the linear
    order of its children, and construct the formula~$\chi_\mu$.
  \end{enumerate}

  Step~1 can be done in~$O(n)$ time by Lemma~\ref{lemma:qoutientGraphs}. 
  
  For Step~2, note that each active node is a least common ancestor of two leaves in~$U_\mu$.
  While it is easy to get all active nodes as least common ancestors, getting
  the edges of $S_\mu$ requires more work.
  Observe that the DFS on~$T$ visits the nodes of~$S_\mu$ in the same order as a DFS on~$S_\mu$.
  Consider~$S_\mu$ embedded such that the children of each node are ordered from left to right by their discovery-times. This also orders the leaves from left to right by their discovery-times. 
  Let~$\lambda$ be an inner node of~$S_\mu$. Let~$\lambda_1$,\,$\lambda_2$
  be two neighboring children of~$\lambda$ with~$\lambda_1$ to the left
  of~$\lambda_2$. Then~$\lambda$ is the least common ancestor of the rightmost
  leaf in~$L(\lambda_1)$ and the leftmost leaf in~$L(\lambda_2)$.
  Hence, each node of~$S_\mu$ is a least common ancestor for a consecutive pair of
  leaves.
 
  We add for every node~$u$ in a set~$U_\mu$ a tuple~$(\mu, u.d)$
  to an initially empty list~$L$.
  We then sort the tuples in~$L$ in linear time using radix sort~\cite{cormen2009introduction}.
  In the sorted list, for every~$\mu \in K(B)$ all tuples~$(\mu, u.d)$ are
  consecutive and the consecutive sublist is sorted by discovery time.

  For~$\mu \in K(B)$ let~$L_\mu$ be a list containing the vertices in~$U_\mu$
  ordered by their discovery time which we get directly from the consecutive
  sublist of~$L$ containing the tuples corresponding to~$\mu$.
  For every pair~$u, v \in U_\mu$ adjacent in~$L_\mu$ we compute~$\lambda =
  \lca_T(u,v)$ using the lowest-common-ancestor data structure for static trees
  by Harel and Tarjan~\cite{harel1984fast}   and insert~$\lambda$ into~$L_\mu$
  between~$u$ and~$v$.
  For a vertex~$u \in U_\mu$ its parent in~$S_\mu$ is the neighbor in~$L_\mu$
  that has a lower position in~$T$. 
  Note that~$u$ is a descendent in~$T$ of all its neighbors in~$L_\mu$. Hence
  if~$u$ has two neighbors in~$L_\mu$ one of them is a descendent of the other.
  Thus the parent of~$u$ in~$S_\mu$ is the neighbor with the higher discovery
  time.  Now we remove all vertices in~$U_\mu$ and possible duplicates of the
  remaining nodes from~$L_\mu$.
  Note that still every~$\lambda$ is a descendent in~$T$ of its neighbors
  in~$L_\mu$.
  Hence we iteratively choose a node~$\lambda$ in~$L_\mu$ whose discovery time
  is higher than the discovery time of its neighbors, compute its parent
  in~$S_\mu$ by comparing the discovery times of its neighbors with each other
  and remove~$\lambda$ from~$L_\mu$.
  
  In Step~3, we turn each active node that stems from a complete node
  into a P-node and each active node that stems from a prime node into
  a Q-node.  For a Q-node~$q$ that stems from a prime node~$\nu$, we
  determine the linear order of its children as follows.  Take the
  set~$X$ of vertices of~$H[\nu]$ that correspond to children of~$\mu$
  in~$S_\mu$, determine the orientation of the complete graph on~$X$
  induced by~$D_\nu$ and sort it topologically.  In total this
  take~$O(n+m)$ time for all active nodes in all PQ-trees.  Using the information
  computed up to this point, it is straightforward to output the
  formula~$\chi$.
\end{proof}

Finally, we combine the constraints from the complete nodes with the
constraints from the prime nodes by setting~$\varphi_{T}= \psi \land \chi$.
The formula~$\varphi_T$ allows us to describe a restricted set of transitive
orientations of~$G$.
We define~$S_T = \{S_\mu \mid \mu \in K(B)\}$.
The canonical modular decomposition~$(B, S_T, \varphi_{T})$ of~$H$ where every complete
node is labeled with the corresponding PQ-tree together with~$\varphi_{T}$ we
call a \emph{constrained} modular decomposition.

We say that a transitive orientation $O$ of $H$ \emph{induces a variable
assignment satisfying $\varphi_T$} if it induces an assignments of the
variables corresponding to prime nodes in $B$ and $Q$-nodes such that
$\varphi_T$ is satisfied for an appropriate assignment for the variables
corresponding to prime nodes in $T$.
Let~$\TO(B, S_T, \varphi_T)$ denote the set containing all transitive
orientations $O\in \TO(B)$ where for every
complete node~$\mu\in B$ the order $O_{\downarrow \mu}$ corresponds to a total
order represented by~$S_\mu$ and that induces a variable assignment that
satisfies~$\varphi_T$.

\subsection{Correctness}
\label{sec:correctness}

We now show that~$\TO(B, S_T, \varphi_T) = \TO(T, D)$. To this end we use that
Lemma~\ref{lemma: prime2} allows to find for an edge~$uv$ that is represented in a
prime node~$\mu$ of~$B$ the prime node~$\mu'=\lca_T(L(\mu))$ of~$T$ where it is represented.
This allows us to establish the identity of certain nodes. The following lemma
does something similar for complete nodes.

\begin{lemma}
\label{lemma: schnitt}
  Let~$uv$,~$wx$ be edges of~$H$ represented in complete nodes~$\mu$,~$\nu$ of~$B$
  and by the same edge in a complete node of~$T$. Then~$\mu=\nu$.
\end{lemma}

\begin{proof}
  Let~$\omega'$ be the node in~$T$ with~$\repT{uv}\in H[\omega']$.
  Assume~$\mu \neq \nu$. Then one of them is the ancestor of the other or they
  are both distinct from~$\lambda = \mathrm{lca}_B(\mu, \nu)$.
  First consider the case that~$\mu$ is an ancestor of~$\nu$.
  Then~$\mu$ has a child~$\mu_1$ such that~$L(\nu) \subseteq L(\mu_1)$. Note
  that~$\repMu{u} \neq \repMu{v}$, hence we have~$\repMu{u} \neq \mu_1$
  or~$\repMu{v} \neq \mu_1$.
  Without loss of generality assume~$\repMu{u} \neq \mu_1$.
  Since~$u\in L(\omega')\cap L(\mu)\setminus L(\mu_1)$ we have~$L(\mu_1)\subseteq L(\omega')$ by Lemma~\ref{lemma: rotBlau} and analogously
  it follows that~$L(\repMu{u})\subseteq L(\omega')$.
  Since $\repT{wx} = \repT{uv} \in H[\omega']$ it is $\omega' = lca_T(L(\mu_1)$.

  Assume that~$\mu_1$ is prime. Then by Lemma~\ref{lemma: prime2},~$\omega'$
  is prime which is a contradiction to the assumption that ~$\omega'$ is
  complete.  Hence~$\mu_1$ must be complete but as~$B$ is the modular
  decomposition of~$H$, no two adjacent nodes in~$B$ are complete.
  Thus~$\mu$ is not an ancestor of~$\nu$ and
  analogously we get that~$\nu$ is not an ancestor of~$\mu$.
  
  It remains to consider the case that~$\nu \neq \lambda \neq \mu$.
  Let~$\lambda_1$ be the child of~$\lambda$ such that~$L(\mu) \subseteq
  L(\lambda_1)$ and let~$\lambda_2$ be the child of~$\lambda$ such that~$L(\nu)
  \subseteq L(\lambda_2)$.
  Again~$\lambda$ has to be complete since otherwise by Lemma~\ref{lemma:
  prime2}~$\omega'$ would be prime which is a contradiction.
  By Lemma~\ref{lemma: rotBlau} we have that~$L(\lambda_1) \cup L(\lambda_2)
  \subseteq L(\omega')$.
  
  Assume that~$\lambda_1$ is prime. Then by Lemma~\ref{lemma:
  prime2},~$\omega'$ is prime which leads to a contradiction.
 Hence~$\lambda_1$ must be complete but again as~$B$ is the modular
 decomposition of~$H$, no two adjacent nodes in~$B$ are complete.
\end{proof}

\begin{theorem}
  \label{theo: singleInput}
  Let~$B$ be the canonical modular decomposition for a graph~$H$ and let~$T$ be a
  restricted modular decomposition for~$H$. 
 Then~$\TO(B, S_T, \varphi_T) = \TO(T, D)$ and~$\TO(B, S_T, \varphi_T)$ can be
 computed in time that is linear in the size of~$H$.
\end{theorem}

\begin{proof}
   
  Let~$O_H \in \TO(T, D)$ and let~$\vec B[H]$ be the orientation of the quotient graphs of~$B$ 
  inducing~$O_H$.
  Then we have already seen that it is necessary that
  every complete node~$\mu$ in~$B$
  is oriented according to a total order represented by~$S_\mu$ 
  and that $O_H$ induces a variable assignment that satisfies~$\varphi_T$.
  Hence~$O_H \in \TO(B, S_T, \varphi_T)$.

  Conversely, let~$O_H \in \TO(B, S_T, \varphi_T)$  and assume~$O_H \notin \TO(T, D)$.
  Then~$O_H$ is either not represented by~$T$ or does not induce~$D$.
  I.e.,~$O_H$ contains two directed
  edges~$uv$,~$wx$ with~$\repT{uv}$  and~$\repT{wx}$ in the same quotient
  graph~$H[\omega']$, such that~$\repT{uv} = \repT{xw}$, or~$\omega'$ is prime
  and~$\repT{uv} \in D_{\omega'}$ but~$\repT{wx} \in D_{\omega'}^{-1}$. 
  Note that if~$\omega'$ is prime, then the first case implies the second one.
  Let~$\mu = \mathrm{lca}_B(u, v)$ and~$\nu = \mathrm{lca}_B(w, x)$.
  We distinguish cases based on the types of~$\mu$ and~$\nu$.
  Let~$\mu'= \lca_T(L(\mu))$,~$\nu'= \lca_T(L(\nu))$ and let~$\vec B[H]$ be the
  orientation of the quotient graphs of~$B$ that induces~$O_H$. 
  Without loss of generality, assume~$\repB{uv} \in D_{\mu}$ and~$\repB{wx} \in
  D_{\nu}$.
  
  \noindent\textbf{Case 1:~$\mu$ and~$\nu$ are both prime.}
  By Lemma~\ref{lemma: prime2} we have that~$\mu' = \omega' = \nu'$ is prime.
  By construction,~$\varphi_T$ enforces that~$uv$,~$wx$ are either represented
  in~$D_{\omega'}$ or in~$D^{-1}_{\omega'}$. Hence, this case does not occur.
  
  \noindent\textbf{Case 2:~$\mu$ is prime and~$\nu$ is complete.}
  By Lemma~\ref{lemma: prime2} we have that~$\mu' = \omega'$ is prime.
  Let~$U$ be a~$\nu$-set containing~$w$ and~$x$. 
  Since~$\repT{wx}\in H[\omega']$, node~$\omega'$ is active with respect to~$U$.
  Hence the PQ-tree~$S_\nu$ contains a Q-node~$q$ that stems from~$\omega'$. 

  By construction~$\varphi_{T}$ contains the constraints~$(x_{\omega'}
  \leftrightarrow x_q)$ and~$(x_{\omega'} \leftrightarrow x_\mu)$ but~$\vec{B}[H]$
  induces~$x_{\mu} = \texttt{true}$,~$x_{q} = \texttt{false}$.
  Hence the variable assignment induced by~$\vec{B}[H]$ does not
  satisfy~$\varphi_{T}$ and thus~$O_H \notin \TO(B, S_T, \varphi_T)$.
  
  \noindent\textbf{Case 3:~$\mu$ is complete and~$\nu$ is prime.}  
  Similar to Case 2.
  
  \noindent\textbf{Case 4:~$\mu,\nu$ are both complete.}
  Here we further distinguish two subcases depending on the
  type of~$\omega'$.
  First assume that~$\omega'$ is prime. 
  Let~$V_1'$ be a~$\mu$-set containing~$u$,~$v$ and 
  let~$V_2'$ be a~$\nu$-set containing~$w$,~$x$. 
  Since~$\repT{uv}$ and~$\repT{wx}$ are edges in~$H[\omega']$, node~$\omega'$
  is active with respect to both~$V_1'$,~$V_2'$.
  Hence~$S_\mu$,~$S_\nu$ contain Q-nodes~$q_1$,~$q_2$ stemming from~$\omega'$. %
  By construction~$\varphi_{T}$ contains the constraints~$(x_{\omega'}
  \leftrightarrow x_{q_1})$ and~$(x_{\omega'} \leftrightarrow x_{q_2})$,
  but~$\vec{B}[H]$ induces~$x_{q_1} = \texttt{ true}$, ~$x_{q_2} = \texttt{
  false}$.
   Hence the variable assignment induced by~$O_H$ does not
   satisfy~$\varphi_{T}$ and thus~$O_H \notin \TO(B, S_T, \varphi_T)$.

   It remains to consider the case that~$\omega'$ is complete.
   By Lemma~\ref{lemma: schnitt} we have~$\mu = \nu$.
   Since~$\omega'$ is not prime, we must have~$\repT{uv} = \repT{xw}$ by assumption.
   Let~$U$ be a~\mbox{$\mu$-set}. 
   Since~$\repT{uv} = \repT{xw}\in H[\omega']$, node~$\omega'$ is 
   active with respect to~$U$ and~$\repNode{\omega'}{u} =
   \repNode{\omega'}{x}$,~$\repNode{\omega'}{v} = \repNode{\omega'}{w}$.
   Hence~$S_\mu$ contains a P-node that stems from~$\omega'$ and demands a
   total order of the children of~$\mu$ where~$\repMu{u}$,~$\repMu{x}$ are either
   both smaller than both~$\repMu{v}$,~$\repMu{w}$, or~$\repMu{u}$,~$\repMu{x}$
   are both greater than both~$\repMu{v}$,~$\repMu{w}$.
   Since~$\vec{B}[H]$ induces~$\repMu{u} < \repMu{v}$ but~$\repMu{x} >
   \repMu{w}$,~$H[\mu]$ is not oriented according to a total order represented
   by~$S_\mu$ and thus~$O_H \notin \TO(B, S_T, \varphi_T)$.
  
   By Lemmas~\ref{lemma: phiprime} and \,\ref{lemma: linear-time PQ} the
   formula~$\varphi_T=\psi\land\chi$ and~$S_T$ can be computed in linear time. 
\end{proof}

  Let~$T$ be a modular decomposition of a graph~$G$ with an induced subgraph~$H$.
  Let~$B$ be the canonical modular decomposition of~$H$ and let~$T|_H$ be the restricted
  modular decomposition for~$H$ we get from~$T$.
  From Lemma~\ref{lemma: mdS} and Theorem~\ref{theo: singleInput} we directly get
  the following corollary.

\begin{corollary}
  The set~$\TO(B, S_{T|_H}, \varphi_{T|_H})$ contains exactly those
  transitive orientations of~$H$
  that can be extended to a transitive orientation of~$G$.
\end{corollary}
Let~$(B, S^1, \varphi_1), \dots, (B, S^r, \varphi_r)$ be constrained modular decompositions
for~$H$. Let~$S_\mu = \bigcap_{i = 1}^r S_\mu^i$ and~$S = \{S_\mu \mid \mu \in K(B)\}$.
The intersection~$(B, S, \varphi)$ of~$(B, S^1, \varphi_1), \dots,(B, S^r, \varphi_r)$ is
the constrained modular decomposition of~$H$ where every complete node
$\mu$ is labeled with the PQ-tree~$S_\mu$ equipped with the 
2-\textsc{Sat}-formula~$\varphi =(\bigwedge_{i = 1}^r \varphi_i) \land
\varphi'$ where~$\varphi'$ synchronizes the Q-nodes in the~$S_\mu$'s with the
Q-nodes in the~$S^i_\mu$'s as follows.
Recall that for every Q-node~$q$ in a tree~$S_\mu^i$ there exists a unique
Q-node~$q'$ in~$S_\mu$ that contains~$q$; either forward or backward.
For every~$i\in\{1,\dots,r\}$, every~$\mu \in K(B)$ and every Q-node~$q$ in~$S_\mu^i$ we
determine the Q-node~$q'$ in~$S_\mu$ that contains~$q$ and add the clause~$(x_q
\leftrightarrow x_{q'})$ if~$q$ has its forward orientation in~$q'$ and~$(x_q
\not \leftrightarrow x_{q'})$ otherwise.

\begin{lemma}
\label{lemma:schnitt}
It is~$\TO(B, S, \varphi) = \bigcap_{i = 1}^r \TO(B, S^i, \varphi_i)$ and we
can compute~$(B, S, \varphi)$ in linear time from~$\{(B, S^i, \varphi_i) \mid 1
\leq i \leq r\}$. 
\end{lemma}

\begin{proof}
 Let~$O_H \in \bigcap_{i = 1}^r \TO(B, S^i, \varphi_i)$ and let~$\vec{B}[H]$ be
 the orientation of the quotient graphs of~$B$ that induces~$O_H$.
 Then for every~$i\in\{1,\dots,r\}$ the variable assignment induced by~$O_H$
  satisfies~$\varphi_i$ and for every complete node~$\mu$ in~$B$, the total
  order~$O_{H\downarrow \mu}$ is represented by~$S_\mu^i$, hence it is
  also represented by~$S_\mu$. 
  Let~$i \in \{1, \dots, r\}$ and let~$q$ be a Q-node in~$S_\mu^i$ and let~$q'$
  be the Q-node in~$S_\mu$ containing~$q$.
  Without loss of generality assume~$q'$ contains~$q$ forwards and~$\vec{B}[H]$
  induces~$x_{q} = \mathtt{true}$ with respect to~$(B, S^i, \varphi^i)$.
  Then~$\vec{B}[H]$ induces~$x_{q'} = \mathtt{true}$ with respect to~$(B, S,
  \varphi)$ and~$\varphi'$ contains the clause~$(x_{q'} \leftrightarrow
  x_q)$ for $x_q$.
  Hence~$\varphi'$ is satisfied in any extension of an assignment of variables
  corresponding to Q-nodes in~$S_\mu,S_\mu^{1},\dots,S_\mu^r$ induced
  by~$\vec{B}[H]$.  Hence~$O_H$ induces a variable assignment that
  satisfies~$\varphi$ and thus~$O_H \in \TO(B, S, \varphi)$. 
 
  Conversely let~$O_H \in \TO(B, S, \varphi)$ and let~$\vec{B}[H]$ be the
  orientation of the quotient graphs of~$B$ that induces~$O_H$. 
  Then for every complete node~$\mu\in B$, the total order~$O_{H\downarrow\mu}$
  is represented by~$S$ and thus by~$S^i_\mu$ for~$i \in \{1, \dots, r\}$.
  Further~$\vec{B}[H]$ induces assignments of the variables corresponding to the
  prime nodes in~$B$ and to the Q-nodes in~$S$. These induced variable
  assignments can be extended to a solution of~$\varphi$.  Let~$i \in \{1,
  \dots, r\}$ and let~$q$ be a Q-node in~$S_\mu^i$ such that~$q'$ is the Q-node
  in~$S_\mu$ containing~$q$.
  Without loss of generality assume~$q'$ contains~$q$ forwards and~$\vec{B}[H]$
  induces~$x_{q'} = \mathtt{true}$ with respect to~$(B, S, \varphi)$. 
  Then~$\varphi$ contains the clause~$(x_{q'} \leftrightarrow x_q)$ and
  hence~$x_{q} = \mathtt{true}$ in any solution of~$\varphi$ that is an
  extension of an assignment of variables corresponding to Q-nodes in~$S_\mu$
  induced by~$\vec{B}[H]$.
  Since for every complete node~$\mu$ in~$B$, the total
  order~$O_{H\downarrow\mu}$ is also represented by~$S^i_\mu$,~$\vec{B}[H]$ also
  induces assignments of the variables corresponding to Q-nodes in~$S_\mu^i$.
  Since~$q'$ contains~$q$ forwards~$\vec{B}[H]$ must induce ~$x_{q} =
  \mathtt{true}$ as well with respect to~$(B, S^i, \varphi_i)$.
  Hence the variable assignment induced by~$\vec{B}[H]$ and~$S_\mu^i$
  satisfies~$\varphi_i$ and thus~$O_H \in \bigcap_{i = 1}^r \TO(B, S^i, \varphi_i)$.  
  
  It remains to show the linear runtime. 
  Let~$G_1 = (V_1, E_1), \dots, G_r = (V_r, E_r)$ with~$n_i
  = |V_i|$ and~$m_i = |E_i|$ for all~$1 \leq i \leq r$ and let~$n = \sum_{i
  = 1}^r n_i$,~$m = \sum_{i = 1}^r m_i$. 
  Since every~$S^i_\mu$ for a node~$\mu\in K(B)$ has one leaf per child
  of~$\mu$, by Lemma~\ref{lemma:pqIntersection} their intersection~$S_\mu$ can
  be computed in~$O(r \cdot |\deg(\mu)|)$ time.  
  Hence in total we need~$\sum_{\mu \in K(B)} 
  O(r \cdot |\deg(\mu)|) = O(n)$ time to compute~$S$.
  For every~$i$, by Theorem~\ref{theo: singleInput}, we can compute~$\varphi_i$
  in~$O(n_i + m_i)$ time.
  By Lemma~\ref{lemma:pqIntersection} we can also find out in~$O(n)$ time which
  Q-nodes are merged and in which direction and hence the construction
  of~$\varphi$ in total takes~$O(n+m)$ time.
\end{proof}

Now consider the case where~$(B, S^1, \varphi_1), \dots, (B, S^r, \varphi_r)$
are the constrained modular decompositions we get from the restricted modular
decompositions~$T_1|_H,\dots,T_r|_H$. By Lemma~\ref{lemma:
schnitt} and Theorem~\ref{theo: singleInput} we have~$\TO(B, S, \varphi) =
\bigcap_{i = 1}^r \TO(B, S_i, \varphi_i)
= \bigcap_{i = 1}^r \TO(T_i, D)$ and by Lemma~\ref{lemma: mdS}~$G_1, \dots, G_r$
are simultaneous comparability graphs if and only if~$\varphi$ is satisfiable
and~$S$ does not contain the null tree.

\begin{theorem}
  \simor can be solved in linear time.
\end{theorem}

\begin{proof}
  Let~$G_1 = (V_1, E_1),\dots, G_r = (V_r, E_r)$ be~$r$-sunflower graphs
  with~$n_i = |V_i|$ and~$m_i = |E_i|$ for all~$1 \leq i \leq r$ and let~$n
  = \sum\limits_{i = 1}^r n_i$,~$m = \sum\limits_{i = 1}^r m_i$.
  We solve~\simor as follows.
  \begin{enumerate}
    \item Compute the canonical modular decomposition~$T_i$ for every~$G_i$ and the canonical modular decomposition~$B$ of~$H$ in ~$O(n+m)$ time by McConnell and Spinrad~\cite{mcconnell1999modular}.
    \item Compute~$T_i|_H$ for every~$i$ in~$O(n)$ time in total.
    \item Compute~$(B, S^i, \varphi_i)$ for every~$i$, in~$O(n +m)$ time in total by Theorem~\ref{theo: singleInput}.
    \item Compute~$(B, S, \varphi)$ in linear time by Lemma~\ref{lemma: schnitt}.
    
    \item Check whether~$S$ contains the null tree and whether~$\varphi$ is satisfiable in linear time. 
  \end{enumerate}
  
  We execute Step 5 as follows.  For~$i\in\{1,\dots,r\}$,~$\varphi_i$
   contains one variable and one constraint per prime node in~$B$, one variable
   per prime node in~$T_i|_H$ and one variable and one constraint per Q-node
   in~$S_\mu^i$.
  Since~$S_\mu^i$ has~$O(n_i)$ nodes, it contains~$O(n_i)$ Q-nodes.
  Hence in total every~$\varphi_i$ contains~$O(n_i)$ variables and clauses. 
  Note that~$\varphi'$ contains one clause per Q-node in the PQ-trees in~$S$. 
  Hence~$\varphi'$ contains~$O(n)$ clauses and variables and thus the 2-SAT
  formula~$\varphi$ can be solved in~$O(n)$ time by Aspvall et
  al.~\cite{aspvall1979linear}.  
  
   If~$S$ does not contain the null tree and~$\varphi$ has a solution, we get
   simultaneous transitive orientations of~$G_1, \dots, G_r$ in linear time by
   proceeding as follows. We orient every complete quotient graph of~$B$
   according to a total order induced by the corresponding PQ-tree where
   every~$Q$-node is oriented according to the solution of~$\varphi$. For a
   prime quotient graph~$H[\mu]$ we choose~$D_\mu$ if in the chosen solution
   of~$\varphi$ we have~$x_\mu = \texttt{true}$ and~$D_\mu^{-1}$ otherwise.
   Together, all these orientations of quotient graphs of~$B$ induce a
   transitive
   orientation on~$H$ and by applying the linear-time algorithm from
   Section~\ref{sec:ParComp} to solve~$\orext$ 
   we can extend it to a transitive orientation of~$G_i$ for every~$i \in \{1,
  \dots, r\}$.
\end{proof}

\section{Permutation Graphs}
\label{sec:perm}

We give algorithms that solve $\repext(\perm)$ and $\sunflower(\perm)$ in linear time using modular 
decomposition and the results from Sections~\ref{sec:ParComp} and~\ref{sec: simComp}. 
To do so, we need the following definitions and observations.
Let~$G = (V, E)$ be a permutation graph and let~$D$ be a representation of~$G$. We denote the upper horizontal line of~$D$ by~$L_1$ and the lower line by~$L_2$. 
\begin{proposition}[{\cite{montgolfier2003decomposition}}]
\label{prop:mD1}
Let~$G= (V, E)$ be a permutation graph and let~$T$ be the canonical modular decomposition of~$G$.
Then for every representation~$D$ of~$G$ and for every~$\mu \in T$, the vertices in~$L(\mu)$ appear consecutively along both~$L_1$ and~$L_2$.
\end{proposition}

Let~$G$ be a permutation graph and let~$T$ be the canonical modular decomposition of~$G$.
We use Proposition~\ref{prop:mD1} to show that for node~$\mu \in T$
we can compute a permutation diagram
of~$G[L(\mu)]$ by iteratively replacing
a line segment representing a descendent~$\nu$
of~$\mu$ in~$T$ by a permutation diagram 
of~$G[\nu]$.
\begin{theorem}
\label{theo: permConstruction}
  Let~$G$ be a permutation graph and let~$T$ be the canonical modular decomposition of~$G$.
  There is a bijection~$\phi$ between the 
  permutation diagrams of~$G$ and the choice of a permutation diagram~$D_\mu$ for each quotient graph~$\mu \in T$.
  Both~$\phi$ and~$\phi^{-1}$ can be computed in~$O(n)$ time.
\end{theorem}

\begin{proof}
  To compute a permutation diagram of~$G$ 
  from~$\{ D_\mu \mid \mu \in T\}$ we traverse~$T$ bottom-up and compute for 
  every~$\mu \in T$ a permutation diagram
  representing~$G[L(\mu)]$ as follows.
  For every child~$\nu$ of~$\mu$ replace the line segment representing~$\nu$ in~$D_\mu$ by the permutation diagram of~$G[L(\nu)]$.
 For every permutation diagram we store 
 two doubly-linked lists containing the order of labels along the two horizontal 
 lines, respectively. Then the replacements described above take linear time in total. 

Conversely, assume that a permutation diagram~$D$ for~$G$ is given as two double linked
lists~$l_1(D)$ and~$l_2(D)$ containing the order of labels along the two horizontal lines, respectively.
We traverse~$T$ bottom-up and consider all leaves as initially visited.
Let~$l_1$ and $l_2$ be two double linked lists and 
initially~$l_1 = l_1(D)$ and~$l_2 = l_2(D)$.
When visiting a non-leaf node~$\mu \in T$ we compute the two double linked lists~$l_1(\mu)$ and~$l_2(\mu)$ 
representing its
permutation diagram as follows.
As an invariant we claim that 
at any time~$l_1$ and~$l_2$ are sublists of~$l_1(D)$
and~$l_2(D)$ and represent the permutation diagram of a 
subgraph of~$G$.
Further~$l_1$ and~$l_2$ contain for every visited child~$\nu$ of an unvisited node 
exactly one entry corresponding to a leaf 
in~$L(\nu)$. 
Clearly the invariant holds at the beginning of the traversal.
Now assume we visit a non-leaf node~$\mu$ and 
the invariant holds.
Let~$l_1(\mu)$ and~$l_2(\mu)$ be the sublist of~$l_1$ and $l_2$ consisting of
all entries corresponding to leaves in~$L(\mu)$, respectively.
The invariant together with Proposition~\ref{prop:mD1} gives 
us that when visiting a node~$\mu$,~$l_1(\mu)$ and~$l_2(\mu)$
are consecutive sublists of~$l_1$ and $l_2$.
To compute~$l_1 (\mu)$ and~$l_2 (\mu)$ we start at an arbitrary child~$\nu$ of~$\mu$
in~$l_1$ and~$l_2$ and search for the right and left end of~$l_1 (\mu)$ and~$l_2 (\mu)$, i.e. on both sides we
search for the first entry that does not correspond to a child of~$\mu$. Finally we remove all list entries corresponding 
to a leaf
in~$L(\mu)$ that is not in~$L(\nu)$ in~$l_1$ and~$l_2$. 
After this step the invariant still holds.
By definition of the quotient graphs and since they contain exactly
one representative per child of~$\mu$,~$l_1(\mu)$ and~$l_2(\mu)$
represent the permutation diagram for $G[\mu]$.
In total these computations take linear time since for every~$\mu \in T$ the
number of visited list entries is in~$O(\deg(\mu))$.
\end{proof}

\subsection{Extending Partial Representations}
\label{sec:ParPerm}
For solving~$\repext(\perm)$ efficiently, we exploit that a given partial representation $D'$
of a permutation graph $G$ is extendible if and only if, 
for every prime node $\mu$ in the canonical modular decomposition of $G$, the partial representation of $G[\mu]$ induced by $D'$ is extendible. Since $G[\mu]$ has only a constant number of representations this can be checked in linear time.

In the following, let~$G = (V, E)$ be a permutation graph and let~$D'$ be a corresponding partial representation of a subgraph~$H = (V', E')$ of~$G$.
Furthermore, let~$T$ be the canonical modular decomposition  for~$G$ and
let~$\mu$ be a node in~$T$.
Let~$U'_\mu \subseteq V'$ be a (not necessarily maximal)~$\mu$-set containing as many vertices in~$V'$ as possible.
Let~$D'_\mu$ be the partial representation of~$G[\mu]$ we get from~$D'$ by removing all line segments corresponding to vertices not in~$U'_\mu$ and
replacing every label u by the label~$\repMu{u}$.
Let~$U_\mu$ be a maximal~$\mu$-set with~$U'_\mu \subseteq U_\mu$.

\begin{lemma}
\label{lemma:mD2}
Let~$D'$ be a partial representation of~$G$. Then~$D'$ can be extended to a representation~$D$ of~$G$ if and only if for every inner node~$\mu$ in~$T$,~$D'_\mu$ can be extended to a representation~$D_\mu$ of~$G[\mu]$.
\end{lemma}

\begin{proof}
  Assume that~$D'$ can be extended to a representation~$D$ of~$G$.
  Then for every inner node~$\mu$ in~$T$,~$D$ restricted to~$U'_\mu$ where every label u is replaced by the label~$\repMu{u}$
  is an extension of~$D_\mu'$ representing~$G[\mu]$.
  
  Conversely assume that for every inner node~$\mu$ in~$T$,~$D'_\mu$ can be extended to a representation~$D_\mu$ of~$G[\mu]$.
  By Theorem~\ref{theo: permConstruction} we get a permutation diagram~$D$ representing~$G$ from the~$D_\mu$s. We show that~$D$ extends~$D'$.
  
  Let~$L_1$ and~$L_2$ denote the upper and the bottom line of~$D$, respectively.
  Now for every pair~$u, v \in V'$ let~$\mu = \lca_T(u, v)$.
  By Proposition~\ref{prop:mD1}~$L(\mu)$ appears consecutively along~$L_1$ and~$L_2$. Hence independently of the choice of~$U'_\mu$,  
  in~$D_\mu$~$\rep{u}$ and~$\rep{v}$ appear in the same order as the labels corresponding to~$u$ and~$v$ in~$D'$ along both horizontal lines.
  Thus~$D$ restricted to~$V'$ coincides with~$D'$.
\end{proof} 

\begin{theorem}
\label{theorem:REPEXT(PG)}
$\repext(\perm)$ can be solved in linear time. 
\end{theorem}

\begin{proof}
Let~$G = (V, E)$ be a permutation graph with~$n$ vertices and~$m$ edges and let~$D'$ be a permutation diagram of an induced subgraph~$H = (V', E')$ of~$G$. We compute the canonical modular decomposition~$T$ of~$G$ in~$O(n+m)$ time~\cite{mcconnell1999modular}.
By Lemma~\ref{lemma:mD2}, it suffices to check whether~$D_\mu'$ can be extended to a representation of~$G[\mu]$ for every~$\mu \in T$.

If~$G[\mu]$ is empty or complete, we can easily extend~$D_\mu'$
to a representation~$D_\mu$ of~$G[\mu]$.
If~$\mu$ is prime, each of~$G[\mu]$ and~$\overline{G}[\mu]$ has exactly two transitive orientations where one is the reverse of the other.
Hence  
there exist only four permutation diagrams 
representing~$G[\mu]$~\cite{golumbic2004algorithmic}.
Note that given an arbitrary permutation diagram~$D_\mu$ for~$G[\mu]$, we get the other three representations by either reversing the order along both horizontal lines, switching~$L_1$ and~$L_2$, or applying both. Hence we can compute all four representations in linear time and check whether one of them contains~$D_\mu'$.

In the positive case, by Theorem~\ref{theo: permConstruction} we
get a representation~$D$ of~$G$ that extends~$D'$ from the representations of the quotient graphs in linear time.
\end{proof}

\subsection{Simultaneous Representations}
\label{sec: simPerm}

Recall that a graph~$G$ is a permutation graph if and only if~$G$ is a comparability and a co-comparability graph.
To solve $\sunflower(\perm)$ efficiently, we build on the following characterization due to 
Jampani and Lubiw.

\begin{proposition}[\cite{jampani2012simultaneous}]
\label{prop:jampani}
  Sunflower permutation graphs 
  are simultaneous permutation graphs if and only if 
  they are simultaneous comparability graphs and 
  simultaneous co-comparability graphs.
\end{proposition}

Let~$G_1, \dots, G_r$ be sunflower permutation graphs.
It follows from Proposition~\ref{prop:jampani} that we can solve $\sunflower(\perm)$ by applying the algorithm to solve~\simor from Section~\ref{sec: simComp} to~$G_1, \dots, G_r$ and their complements~$\overline{G_1}, \dots, \overline{G_r}$. 
Since the algorithm from Section~\ref{sec: simComp} needs time linear in the size of the input graphs and~$\overline{G}$ may have a quadratic number of edges, this approach takes quadratic time in total. This already improves the cubic runtime of the algorithm presented 
by Jampani and Lubiw~\cite{jampani2012simultaneous}.

We show that it is possible to use the 
algorithm from Section~\ref{sec: simComp} to solve also~\simor for co-comparability graphs and thus~$\sunflower(\perm)$ in linear time.  
Let~$G$ be a permutation graph with induced subgraph~$H$.
Recall that a graph~$G$ and its complement have the same canonical modular decomposition, they only differ in the type of the nodes.
For a node~$\mu$ with~$G[\mu]$ prime also~$\overline{G}[\mu]$ is prime; if~$G[\mu]$ is complete or empty,~$\overline{G}[\mu]$ is empty or complete, respectively~\cite{gallai1967transitiv}.

Let~$B$ be the canonical modular decomposition of~$H$
and let~$T$ be the canonical modular decomposition of~$G$ restricted to~$H$. The goal is to compute constrained modular decompositions~$(B, S^1, \varphi_1), \dots, (B, S^r, \varphi_r)$ such that for every~$i \in \{1, \dots, r\}$~$(B, S^i, \varphi_i)$ contains exactly those transitive orientations of~$\overline{H}$ that can be extended to a transitive orientation of~$\overline{G}_i$, and~$(B, S, \varphi)$ with ~$\TO(B, S, \varphi) = \bigcap_{i = 1}^r \TO(B, S^i, \varphi_i)$.
In linear time we cannot explicitly compute~$\overline{G}$ and the corresponding quotient graphs. 
Hence we cannot store a default orientation
for the prime quotient graphs~$\overline{H}[\mu]$.  We can, however, compute a \emph{default permutation diagram}~$D_\mu$ representing~$\overline H[\mu]$ from its default orientation~$O_\mu$ in linear time~\cite{mcconnell1999modular}.
We apply the algorithm from Section~\ref{sec: simComp}. 
The 2-\textsc{Sat}-formula~$\psi$ describing the dependencies between the
orientations of the prime nodes can be computed in~$O(|V(H)| + |E(H)|)$
time since when it is required to check whether an oriented edge~$uv$ is contained in the default orientation of a prime quotient graph~$\overline{H}[\mu]$, 
we do this in constant time by checking for the non-edge~$uv$ in~$H[L(\mu)]$
whether~$\repMu{u} <_\mu \repMu{v}$. 
We can determine one non-edge for each prime quotient graph of~$B$ in~$O(|E(H)|)$ time in total.
Note that one non-edge per prime quotient graph suffices to answer the queries.
Further we have to compute the PQ-trees for the complete nodes in~$T$ and~$\chi$. Note that the only step in the computation of the PQ-trees that takes potentially quadratic time for~$\overline{H}$ is the computation of the order of the children for the Q-nodes.
But since every prime node is labeled with a default representation, we can compute the PQ-trees and~$\chi$ in~$O(|V(H)|)$ time instead of~$O(|V(\overline{H})|) + |E(\overline{H})|)$ time, since we get the order of the children of a Q-node directly from the corresponding permutation diagram. 

\begin{theorem}
\label{theorem:simrepPG}
    $\sunflower(\perm)$ can be solved in linear time.
\end{theorem}

\begin{proof} 
Let~$\varphi$ be the 2-\textsc{Sat} formula we get for~$H$ and let~$\overline{\varphi}$ be the formula from~$\overline{H}$.
Further let~$S$ and~$\overline{S}$ be the set of 
intersected PQ-trees we get for~$G_1, \dots, G_r$ and~$\overline{G_1}, \dots, \overline{G_r}$, respectively. 
Then
$G_1, \dots, G_r$ are simultaneous permutation graphs if and only if~$\varphi$ and~$\overline{\varphi}$ are satisfiable and
neither~$S$ nor~$\overline{S}$ contain the null tree. This can be checked in linear time.
In the positive case we proceed as follows.
\begin{enumerate}
  \item For every quotient graph in~$B$ compute the permutation diagram induced by the solutions of~$\varphi$ and~$\overline{\varphi}$.
  \item Compute a permutation diagram~$D_H$ representing~$H$ from the representations of the quotient graphs (see Theorem~\ref{theo: permConstruction}).
  \item For every input graph~$G_i$ use the algorithm from Section~\ref{sec:ParPerm} to extend~$D_H$ to a representation of~$G_i$.
\end{enumerate}

More precisely, Step~2 works as follows.
For complete or empty nodes~$\mu$ we 
construct the representation by choosing a 
linear order of their children represented by 
the corresponding PQ-tree where every Q-node is oriented 
according to a solution of~$\varphi$ or~$\overline{\varphi}$, for the 
upper line. If~$H[\mu]$ is empty, we choose 
the same order for the bottom line, if 
$H[\mu]$ is complete, the bottom line is labeled 
with the reversed order.

For a prime quotient graph~$\mu$ we 
distinguish four cases.
Let~$x_\mu$ be the variable encoding the orientation of~$H[\mu]$ and let~$y_\mu$ be the variable encoding the orientation of~$\overline{H}[\mu]$.
If~$x_\mu = \texttt{true}, y_\mu = 
\texttt{true}$
we choose the default representation~$D_\mu$,
if~$x_\mu = \texttt{true}, y_\mu = 
\texttt{false}$
we reverse the orders along both horizontal 
lines of~$D_\mu$, if~$x_\mu = \texttt{false}, 
y_\mu = \texttt{true}$
we switch the orders along the horizontal 
lines of~$D_\mu$ and if~$x_\mu = 
\texttt{false}, y_\mu = \texttt{false}$
we reverse the orders along both horizontal 
lines of~$D_\mu$ and switch them.
\end{proof}

\section{Circular permutation graphs}
\label{sec:cperm}
In this section we give efficient algorithms for solving \repext(\cperm) and \sunflower(\cperm)
based on the linear time algorithms for solving \repext(\perm) and \sunflower(\perm) and
the \emph{switch} operation.

Let~$G = (V,E)$ be a circular permutation graph. \emph{Switching} a vertex~$v$ in~$G$, i.e., connecting it to all 
vertices it was not adjacent to in~$G$ and removing all edges to its former neighbors, gives us the graph~$G_v = (V, 
E_v)$ with~$E_v = \binom{V}{2} \setminus E$.  The graph we obtain by switching all neighbors of a vertex~$v$ we denote by~$G_{N(v)}$.

Let~$C$ be a circular permutation diagram representing a permutation graph~$G = (V, E)$ and let~$v \in V$ be a vertex in 
$G$. The chord~$v$ in~$C$ can be \emph{switched} to the chord~$v'$, if~$v'$ has the 
same endpoints as~$v$, but intersects exactly the chords~$v$ does not intersect in~$C$. Hence, 
this modified circular permutation diagram~$C'$ is a representation of~$G_v$~\cite{rotem1982circular}. 
Here it suffices to know that 
there exists a~$v'$ that~$v$ can be switched to, hence, we use 
the term \emph{switching chord v}, without further specifying~$v'$.
\subsection{Extending Partial Representations}

In this section we show that $\repext(\cperm)$ can be solved via a linear-time Turing reduction to $\repext(\perm)$ using the switch operation.

 Let~$G$ be a circular permutation graph and let~$C_p$ be a circular permutation diagram representing an induced subgraph~$H$ of~$G$. Let~$v$ be an arbitrary vertex in~$G$ and let~$G' = G_{N(v)}$. The circular permutation diagram we obtain by switching all chords corresponding to neighbors of~$v$ in~$C_p$ we denote by~$C_p'$.
Since~$G'$ is a permutation graph~$C_p'$ can be transformed into a permutation diagram~$D_p'$ representing the same induced subgraph. Observe that~$C_p$ is extendible to a representation of~$G$ if and only if~$D'_p$ can be extended to a representation of~$G'$ since an extension of~$D'_p$ can easily be transformed back into a circular permutation diagram extending~$C_p$.
This can be checked in linear time using the algorithm from Section~\ref{sec:ParPerm}.

To achieve a linear runtime in total we have to switch 
the neighborhood of a vertex~$v$ with minimum degree in~$G$.
In case~$v \in V(H)$ we can easily transform~$C_p'$ into a permutation diagram~$D_p'$ representing the same induced subgraph by opening~$C_p'$ along the isolated chord~$v$.
If~$v \notin V(H)$ it is more difficult to find a position where we can open~$C_p'$. 

\begin{theorem}
  $\repext(\cperm)$ can be solved in linear time.
\end{theorem}

\begin{proof}
Let~$G$ be a circular permutation graph and let~$C'$ be a circular permutation diagram representing an induced subgraph~$H$ of~$G$. Let~$v$ be a vertex of minimum degree in~$G$.
Then we compute~$G'$ in linear time~\cite{sritharan1996linear}.
In case~$v \in V(H)$ we open~$C_p'$ along the isolated chord~$v$
and check in linear time whether the resulting permutation diagram can be extended to a representation of~$G'$.
If~$v \notin V(H)$ we proceed as follows.

Let~$H'$ be the graph we obtain from~$H$ by switching all vertices that are adjacent to~$v$ in~$G$ and let~$CC(H')$ be the set of all connected components in~$H'$. Note that in~$C_p'$
the endpoint of chords corresponding to vertices of the same connected component of~$H'$ appear as consecutive blocks along both the outer and the inner circle and that the only positions where we can open~$C_p'$ are \emph{between} these blocks; see Figure~\ref{fig:cc}.
\begin{figure}
  \centering
  \includegraphics{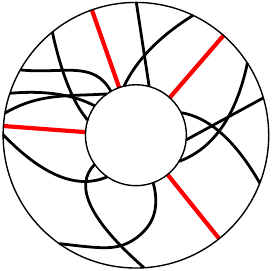}
  \caption{The circular permutation representation~$C$ of a graph with four connected components. The red line segments mark the positions where~$C$ can be opened.}
  \label{fig:cc}
\end{figure}
If~$|CC(H')| = 1$ or~$|CC(H)| = 2$
there are only one or two positions, respectively, where
we can open~$C_p'$. Hence in these cases we can 
construct all one or two possible permutation diagrams~$D_p'$ representing~$H'$ and check whether one of them is extendible.  
Else we distinguish several further cases.

\noindent\textbf{Case 1:~$|CC(H)| > 2$} and there exists a vertex~$u$ in~$G$ such that~$u$ has a neighbor in every 
connected component of~$H'$ but is not adjacent to all vertices in~$H$.
Note that if more than two connected components in~$CC(H')$ contain vertices not adjacent to~$u$,~$C'_p$ is not extendible. The same holds if two non-adjacent connected components in~$CC(H')$ contain vertices not adjacent to~$u$.
Else we distinguish further subcases.

\noindent\textbf{Case 1a:} Two connected components in~$CC(H')$ contain vertices not adjacent to~$u$ and the corresponding blocks are adjacent in~$C_p'$. Then the only position where we can open~$C_p'$ is between these two blocks.

\noindent\textbf{Case 1b:} Only one connected components in~$CC(H')$ contains vertices not adjacent to~$u$. Then we can either open to the left or to the right of the block. In this case we check whether one of the to possibilities gives us an extendible permutation diagram.

\noindent\textbf{Case 2:} there exists no vertex~$u$ in~$G$ that has a neighbor in every connected component of~$H'$ but is not adjacent to every vertex in~$H$. Then we remove all vertices from~$G$ that are adjacent to every other vertex in~$H$.
If~$G$ is empty afterwards, we can choose an arbitrary gap between two adjacent blocks and open~$C_p'$ there.
Else we iteratively merge for each remaining vertex~$u$ all connected components of~$H'$ that contain a vertex adjacent to~$u$. Since~$G$ is connected we end up with either one or two blocks which gives us only one or two positions to open~$C_p'$, respectively. 
We check whether one of them leads to an extendible permutation diagram.
\end{proof}

\subsection{The Simultaneous Representation Problem}
In this section we show that $\sunflower(\cperm)$ can be solved via a quadratic-time Turing reduction to \sunflower(\perm)
using the switch operation. Here we need to switch the
neighborhood of a vertex v shared by all input graphs. Hence in contrast to the reduction in
Section 6.1, where we switched the neighborhood of a vertex of minimum degree, the graph~$G_{N(v)}$ may have quadratic size.

\begin{lemma}
\label{lemma: help}
  Let~$G_1, G_2, \dots, G_r$ be sunflower circular permutation graphs sharing an induced subgraph~$H$. Let~$v$ be a vertex in~$H$ and let~$G_i'$ be the graph we obtain by switching all neighbors of 
  vertex~$v$ in~$G_i$ for~$i \in \{1, \dots, r\}$.
  Then~$G_1, G_2, \dots, G_r$ are simultaneous circular permutation graphs if and only 
  if~$G_1', G_2', \dots, G_r'$ are simultaneous permutation graphs.
\end{lemma}

\begin{proof}
  Recall that~$G_1', G_2', \dots, G_r'$ are indeed permutation graphs~\cite{rotem1982circular}.
  Let~$G_1, G_2, \dots, G_r$ be simultaneous circular permutation graphs.
  Then there exist circular permutation diagrams~$C_1, C_2, \dots, C_r$ such that for 
  every~$1 \leq i \leq r$,~$C_i$ represents~$G_i$ and all~$C_i$s coincide restricted to the vertices of~$H$.
  We get a circular permutation diagram~$C_i'$ representing~$G_i'$ by switching all chords corresponding
  to neighbors of vertex~$v$ in~$C_i$.
  Since the order of the vertices along the inner and outer circle is not affected by the switch operation, we know that all~$C_i'$s coincide restricted to the vertices of~$H$. 
  Recall that after switching all neighbors of chord~$v$ in a circular permutation diagram, no 
  chord intersects~$v$ any more and hence we obtain a permutation diagram~$D_i'$ 
  representing~$G_i'$ by opening~$C_i'$ along the chord~$v$.
  Then the~$D_i'$s also coincide on the vertices of~$H$ and hence~$G_1', G_2', \dots, G_r'$ are simultaneous 
  permutation graphs.
  
  Conversely let~$G_1', G_2', \dots, G_r'$ be simultaneous permutation graphs.
  Then there exist permutation diagrams~$D_1', D_2', \dots, D_r'$ such that for every~$1 \in \{1, \dots, 
  r\}$,~$D_i'$ represents~$G_i'$ and all~$D_i'$s coincide on the vertices of~$H$.
  Recall that we can transform every linear permutation diagram~$D_i'$ into a circular permutation 
  diagram~$C_i'$ representing~$G_i'$.
  Note that the~$C_i'$s also coincide on the vertices of~$H$.
  Now we obtain a circular permutation diagram~$C_i$ representing~$G_i$ by switching all chords that
  correspond to vertices adjacent to~$v$.
  The~$C_i$s still coincide on the vertices of~$H$ since the switch operation 
  does not change the order of the vertices along the outer or the inner circle of a circular 
  permutation diagram and chords corresponding to vertices in~$H$ are either switched in every~$C_i$ or 
  in none of them.
\end{proof}

\begin{theorem}
  $\sunflower(\cperm)$ can be solved in $O(n^2$) time.
\end{theorem}

\begin{proof}
  Let~$G_1, \dots, G_r$ be sunflower circular permutation graphs.
  Let~$v$ be a vertex in~$H$ and let~$G_i'$ be the graph we obtain by switching all neighbors of 
  vertex~$v$ in~$G_i$ for~$i \in \{1, \dots, r\}$.
  By Lemma~\ref{lemma: help} to solve \textsc{SimRep(\cperm)} for~$G_1, G_2, \dots, G_r$ it suffices 
  to solve \textsc{SimRep(\perm)} for~$G_1', G_2', \dots, G_r'$. Computing~$G_1', G_2', \dots, G_r'$ takes quadratic time. 
  In Section~\ref{sec: simPerm} we have seen that \textsc{SimRep(perm)} can be solved in 
linear time for sunflower permutation graphs~$G_1', \dots, G_r'$, hence in total we need
quadratic time to solve $\sunflower(\cperm)$.
  
  If~$G_1, G_2, \dots, G_r$ are simultaneous 
  circular 
  permutation graphs we get corresponding simultaneous representations by transforming simultaneous linear
  permutation diagrams representing~$G_1', G_2', \dots, G_r'$ into circular permutation diagrams and 
  switching all chords corresponding to a neighbor of vertex~$v$ in~$H$. 
  This takes quadratic time in total.
\end{proof}

\section{Simultaneous Orientations and Representations for General Comparability Permutation Graphs}
\label{sec:hard}
We show that the simultaneous orientation problem for comparability graphs and the simultaneous representation problem for permutation graphs are \texttt{NP}-complete in the non-sunflower case.

\begin{theorem}
\label{theo:np}
  \simorNon for~$k$ comparability graphs where~$k$ is part of the input is \texttt{NP}-complete.
\end{theorem}

\begin{proof}
  Clearly the problem is in \texttt{NP}.
  
   To show the \texttt{NP}-hardness, we give a reduction from the known \texttt{NP}-complete problem \textsc{TotalOrdering}~\cite{opatrny1979total}, which is defined as follows. Given a finite set~$S$ and a finite set~$T$ of triples of elements in~$S$, decide 
  whether there exists a total ordering~$<$ of~$S$ such that for all triples~$(x, y, z) \in T$ either 
 ~$x < y < z$ or~$x > y > z$. Let~$H_T$ be an instance of \textsc{TotalOrdering} 
  with~$s = |S|$ and~$t = |T|$. 
  We number the triples in~$T$ with~$1, \dots, t$ and denote the~$i$th triple by 
 ~$(x_i, y_i, z_i)$. 
  We construct an instance~$H_S$ of \simorNon, consisting of undirected 
  graphs~$G_0, \dots, G_t$ as follows.
  
  \begin{itemize}
    \item~$G_0 := (S, E_0)$ is the complete graph with one vertex for each element in~$S$. 
    \item~$G_i := (V_i, E_i)$ for~$1 \leq i \leq t$ is the graph with vertex set~$V_i = \{x_i, y_i, z_i, a_i, b_i \}$, where $a_i,b_i \notin S$ are new vertices, and edges $E_i = \{x_ia_i, x_iy_i, x_iz_i,$~$y_iz_i, z_ib_i \}$; see Figure \ref{fig:G_i}.
    \end{itemize}
  \begin{figure}[t]
    \centering
    \includegraphics{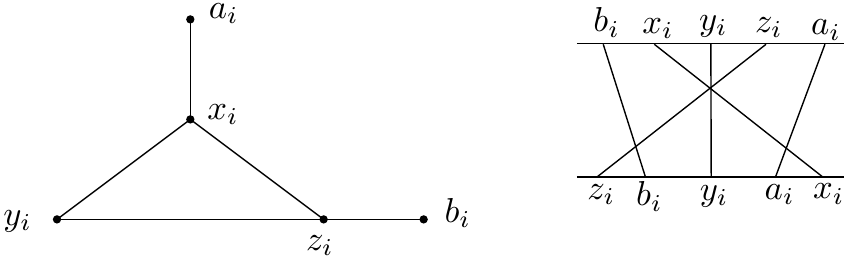}
    \caption{The (permutation) graph~$G_i$ corresponding to the~$i-th$ triple~$(x_i, y_i, z_i)$ of a    
    \textsc{TotalOrdering} instance with corresponding permutation diagram.}
    \label{fig:G_i}
  \end{figure}
  
  For every~$1 \leq i \leq t$, the graph~$G_i$ has  
  exactly two transitive orientations, namely~$\{\overrightarrow{x_ia_i}, 
  \overrightarrow{x_iy_i},$ $\overrightarrow{x_iz_i}, \overrightarrow{y_iz_i}, \overrightarrow{b_iz_i}\}$ 
  and its reversal. 
  We now have to show that~$H_T$ has a solution if and only if~$G_0, \dots, G_t$ are simultaneous 
  comparability graphs. 
  
  First, assume that~$G_0, \dots, G_t$ are simultaneous comparability graphs. Hence there exist 
  orientations~$T_0, \dots, T_t$ such that for every~$0 \leq i \leq t$,~$T_i$ is an 
  orientation of~$G_i$ and for 
  every~$j, k \in \{0, \dots t\}$ every edge in~$E_j \cap E_k$ is oriented in the same way in both~$T_j$ 
  and~$T_k$. Then the orientation of the complete graph~$G_0$ implies a total order on the elements 
  of~$T$, where~$u < v$ if and only if the edge~$uv$ is oriented from~$u$ to~$v$. By construction, there 
  are only two valid transitive orientations for~$G_i$. The one given above implies~$x_i < y_i < z_i$ and 
  for the reverse orientation we get~$z_i < y_i < x_i$. Hence the received total order satisfies that for 
  every triple~$(x, y, z) \in T$ we have either~$x < y < z$ or~$x > y > z$. 
  
  Conversely, assume that~$H_T$ has a solution. Then there exists a total order~$<$ such that for all triples 
 ~$(x, y, z) \in T$ either~$x < y < z$ or~$x > y > z$ holds. We get a transitive orientation~$T_i$ of~$G_i$ for~$1 \leq i \leq t$ by orienting all edges between the vertices~$x_i, y_i$ and~$z_i$ towards 
  the greater element according to the order $<$. If~$x_i$ is the smallest element of the triple, then 
  we choose~$\overrightarrow{x_ia_i}$ and~$\overrightarrow{b_iz_i}$, else~$z_i$ is the smallest element 
  and we choose~$\overrightarrow{a_ix_i}$ and~$\overrightarrow{z_ib_i}$. Finally, we orient the edges in~$G_0$ also towards the greater element according to $<$. This gives us orientations~$T_0, \dots, T_t$ of~$G_0, \dots, G_t$ with the property that for every~$j, k \in \{0, \dots t\}$ every edge 
  in~$E_j \cap E_k$ is oriented in the same way in both~$T_j$ and~$T_k$. 
  Hence~$G_0, \dots, G_t$ are simultaneous comparability graphs. 

  Hence the instance~$H_T$ of \textsc{TotalOrdering} has a solution if and only if the instance~$H_S$ is a yes-instance of    
  \simorNon.  Since~$H_S$ can be constructed from~$H_T$ in polynomial time, it follows that \simorNon is \texttt{NP}-complete.
\end{proof}

\begin{theorem}
  \textsc{SimRep(Perm)} for~$k$ permutation graphs where~$k$ is not fixed is \texttt{NP}-complete.
\end{theorem}

\begin{proof}
  Clearly the problem is in \texttt{NP}.
  
  To show the \texttt{NP}-hardness, we
  use the same reduction as in the proof of Theorem~\ref{theo:np}. 
  Let $H_T$ be the corresponding instance of \textsc{TotalOrdering}. 
  Note that~$G_0, \dots, G_t$ are permutation graphs
  and thus~$\overline{G_0}, \dots, \overline{G_t}$ are comparability graphs.. 
  We already know that~$H_T$ has a solution if and only if~$G_0, \dots, G_t$ are simultaneous comparability graphs.
  Hence it remains to show that~$G_0, \dots, G_t$ are simultaneous co-comparability graphs if~$H_T$ has a solution.
  Note that~$E(\overline{G_0}) = \emptyset$ and for every~$1 \leq i \leq t$ for every edge~$uv \in 
  E(\overline{G_i})$ at least one of the endpoints~$u$ and~$v$ is in~$\{a_i, b_i\}$. 
  Hence,~$\overline{G_0}, \dots, \overline{G_t}$ pairwise do not share any edges and thus they are also 
  simultaneous comparability graphs. 
\end{proof}

\section{Conclusion}
We showed that the orientation extension 
problem and the simultaneous orientation 
problem for sunflower comparability graphs
can be solved in linear time using the 
concept of modular decompositions. 
Further we were able to use these 
algorithms to solve the partial 
representation problem for permutation and 
circular permutation graphs and the simultaneous representation 
problem for sunflower permutation graphs
also in linear time.
For the simultaneous representation problem for circular permutation
graphs we gave a quadratic-time algorithm.
The non-sunflower case of the simultaneous orientation 
problem and the simultaneous representation 
problem turned out to be \texttt{NP}-complete.

It remains an open problem whether the simultaneous representation problem for sunflower circular permutation graphs can be solved in subquadratic time. Furthermore it would be interesting to examine whether the concept of modular decomposition is also applicable to solve the partial representation and the simultaneous representation problem for further graph classes, e.g. trapezoid graphs.
There may also be other related problems that can be solved for comparability, permutation and circular permutation graphs with the concept of modular decompositions.

\bibliography{partial-modular}

\end{document}